\documentclass[12pt]{article}
\usepackage[affil-it]{authblk}
\usepackage{etex}
\usepackage[utf8]{inputenc}
\usepackage[T1]{fontenc}
\usepackage{textcomp}
\usepackage[a4paper,tmargin=3cm,bmargin=3cm,rmargin=2.2cm,lmargin=2.2cm]{geometry}
\usepackage{enumerate}

\usepackage{amsfonts,amssymb,amsmath,amscd,mathtools,cite,slashed,mathrsfs}
\usepackage[amsthm,amsmath]{ntheorem}
\usepackage{graphicx}
\usepackage[all]{xy}
\usepackage{lmodern}

\usepackage{tikz}
\usepackage{tikz-cd} 
\usetikzlibrary{shapes.misc,arrows,decorations.markings}
\usetikzlibrary{matrix}

\usepackage{color}
\usepackage{psfrag}
\usepackage{epsfig}
\usepackage{pdflscape}
\usepackage{hyperref}
\usepackage{verbatim} 
\usepackage{graphicx}

\usepackage{xcolor}
\usepackage{framed}
\definecolor{shadecolor}{gray}{0.8}
\definecolor{lgray}{gray}{0.5}

\newcounter{mnotecount}[section]
\renewcommand{\themnotecount}{\thesection.\arabic{mnotecount}}
\newcommand{\mnote}[1]
{\protect{\stepcounter{mnotecount}}$^{\mbox{\footnotesize
$
\bullet$\themnotecount}}$ \marginpar{
\raggedright\tiny\em
$\!\!\!\!\!\!\,\bullet$\themnotecount: #1} }

\DeclareFontFamily{U}{mathx}{\hyphenchar\font45}
\DeclareFontShape{U}{mathx}{m}{n}{
      <5> <6> <7> <8> <9> <10>
      <10.95> <12> <14.4> <17.28> <20.74> <24.88>
      mathx10
      }{}
\DeclareSymbolFont{mathx}{U}{mathx}{m}{n}
\DeclareFontSubstitution{U}{mathx}{m}{n}
\DeclareMathAccent{\widecheck}{0}{mathx}{"71}
\DeclareMathAccent{\wideparen}{0}{mathx}{"75}

\newtheorem{thm}{Theorem}
\newtheorem{cor}[thm]{Corollary}
\newtheorem{lem}[thm]{Lemma}
\newtheorem{prop}[thm]{Proposition}

\theoremstyle{definition}
\newtheorem{defn}[thm]{Definition}
\newtheorem{rem}[thm]{Remark}
\newtheorem{ex}[thm]{Example}

\newcommand{\sN}{{\mathbb N}}

\newcommand{\sR}{{\mathbb R}}

\newcommand{\sM}{{\mathbb M}}
\newcommand{\A}{\mathcal{A}}
\newcommand{\B}{\mathcal{B}}

\newcommand{\Cf}{\mathscr{C}}

\newcommand{\M}{\mathcal{M}}
\newcommand{\N}{\mathcal{N}}

\newcommand{\Pf}{\mathscr{P}}
\newcommand{\K}{\mathcal{K}}
\newcommand{\cL}{\mathcal{L}}
\newcommand{\T}{\mathcal{T}}

\newcommand{\cS}{\mathcal{S}}
\newcommand{\X}{\mathcal{X}}
\newcommand{\Y}{\mathcal{Y}}
\newcommand{\U}{\mathcal{U}}
\newcommand{\V}{\mathcal{V}}
\newcommand{\W}{\mathcal{W}}

\newcommand{\itGamma}{\mathit{\Gamma}}
\newcommand{\itSigma}{\mathit{\Sigma}}
\newcommand{\itPi}{\mathit{\Pi}}
\newcommand{\loc}{\mathrm{loc}}
\newcommand{\cb}{\mathrm{b}}
\newcommand{\cc}{\mathrm{c}}

\newcommand{\bx}{\textup{\textbf{x}}}

\newcommand{\bv}{\mathbf{v}}

\newcommand{\bA}{\mathbf{A}}
\newcommand{\bB}{\mathbf{B}}

\DeclareMathOperator{\supp}{supp}

\DeclareMathOperator*{\esssup}{ess\,sup}

\newcommand{\ev}{\textnormal{ev}}

\newcommand{\im}{\textnormal{im}}
\newcommand{\vol}{\textnormal{vol}}

\newcommand{\vc}{\vcentcolon =}
\newcommand{\cv}{= \vcentcolon}
\DeclareMathOperator{\id}{id}
\DeclareMathOperator{\dyw}{div} 

\newcommand{\wv}{\widecheck}

\begin{document}


\title{Causal evolution of probability measures\\ and continuity equation}
\author{Tomasz Miller${}^{1}$\thanks{tomasz.miller@uj.edu.pl}}

\affil{\small ${}^1$ Copernicus Center for Interdisciplinary Studies, Jagiellonian University,
\\Szczepa\'nska 1/5, 31-011 Krak\'ow, Poland}
\date{\today}


\maketitle

\begin{abstract}
We study the notion of a causal time-evolution of a conserved nonlocal physical quantity in a globally hyperbolic spacetime $\M$. The role of the `global time' is played by a chosen Cauchy temporal function $\T$, whereas the instantaneous configurations of the nonlocal quantity are modeled by probability measures $\mu_t$ supported on the corresponding time slices $\T^{-1}(t)$. We show that the causality of such an evolution can be expressed in three equivalent ways: (i) via the causal precedence relation $\preceq$ extended to probability measures, (ii) with the help of a probability measure $\sigma$ on the space of future-directed continuous causal curves endowed with the compact-open topology and (iii) through a causal vector field $X$ of $L^\infty_\loc$-regularity, with which the map $t \mapsto \mu_t$ satisfies the continuity equation in the distributional sense. In the course of the proof we find that the compact-open topology is sensitive to the differential properties of the causal curves, being equal to the topology induced from a suitable $H^1_\loc$-Sobolev space. This enables us to construct $X$ as a vector field in a sense `tangent' to $\sigma$. In addition, we discuss the general covariance of descriptions (i)--(iii), unraveling the geometrical, observer-independent notions behind them.
\end{abstract}

MSC classes: 53C50, 53C80, 28E99, 60B05


\section{Introduction and main results}
\label{sec::intro}

Lorentzian optimal transport theory is a new fast-developing field of mathematical physics. Starting with the work of Brenier \cite{Brenier2003}, who was the first to study the Monge--Kantorovich problem with relativistic cost functions, several authors successfully adapted the tools of standard, Riemannian optimal transport theory to the setting of Lorentzian (or even Lorentz--Finsler) manifolds \cite{Bertrand2013, Bertrand2018, Kell2018, Suhr2016}. It is readily motivated by physical considerations, including the so-called `early universe reconstruction problem' \cite{BrenierFrisch2003, Frisch2002, Frisch2006, Frisch2011}, relationships between general relativity and the second law of thermodynamics \cite{McCann18, Suhr18}, and the notion of causality and causal evolution of nonlocal objects \cite{AHP2017, PRA2017, Miller17a}. 

The causality requirement, understood intuitively as the impossibility of superluminal motion or --- in the current context --- of superluminal transportation of probability measures, is arguably a minimal one that must be met by any Lorentzian optimal transport plan. It can be ensured simply by demanding the cost function to be finite only for causally connected pairs of events (as was done already in \cite{Brenier2003}, cf.  \cite{Suhr2016}). This naturally leads to the following extension of the standard causal precedence relation $\preceq$ on a spacetime\footnote{We adopt the metric signature $(-+ \ldots +)$ and set the speed of light equal to 1. For the sake of concreteness, we assume that the spacetime dimension is $1+3$ unless explicitly stated otherwise. See Appendix \hyperref[sec::AppB]{B} for the basic definitions and facts from causality theory.} $\M$ onto the space $\Pf(\M)$ of Borel probability measures on $\M$ \cite{AHP2017, Suhr2016}.
\begin{defn}
\label{causal_rel_def}
For any $\nu_1, \nu_2 \in \Pf(\M)$ we say that $\nu_1$ causally precedes $\nu_2$ (in symbols $\nu_1 \preceq \nu_2$) if there exists a joint probability measure $\omega \in \Pf(\M^2)$ such that
\begin{itemize}
\item $\pi^i_\sharp\omega = \nu_i$, where $\pi^i: \M^2 \rightarrow \M$ is the canonical projection on $i$-th argument, $i=1,2$.
\item $\omega(J^+) = 1$, where $J^+ \vc \{ (p,q) \in \M^2 \, | \, p \preceq q \}$. 
\end{itemize}
Every such $\omega$ is called a \emph{causal coupling}\footnote{Recall that in optimal transport theory a coupling \cite{Villani2008} (or a transport plan \cite{UsersGuide,garling2017polish}) is a joint probability measure satisfiying the first of the two above conditions.} of $\nu_1$ and $\nu_2$. The set of such couplings is denoted by $\itPi_\preceq(\nu_1,\nu_2)$.
\end{defn}

Extended in this manner, relation $\preceq$ can be employed to model a causal time-evolution of a spatially distributed physical entity, e.g. a dust cloud or energy density. In order to motivate the definition, notice first that since the entity in question might be highly nonlocal in space, the standard prescription invoking a local coordinate chart 
 is insufficient. Instead, the description of such an evolution requires some notion of a `global time' and of `global simultaneity hypersurfaces', much like in the FLRW cosmological models \cite{HawkingEllis,BN83} or in fact like in the above-mentioned early work of Brenier \cite{Brenier2003}. This, in turn, sets restrictions on the causal structure of the underlying spacetime manifold. Concretely, the so-called \emph{temporal} functions --- i.e. smooth maps with past-directed timelike gradient serving as the `global time' --- exist precisely in stably causal spacetimes \cite{BS04}. If, on top of it, we demand the `simultaneity hypersurfaces' to be Cauchy hypersurfaces, which guarantees the well-posedness of the Cauchy problem, one arrives at the notion of a globally hyperbolic spacetime \cite{HR09}. The outstanding feature of such spacetimes, first discovered by Geroch \cite{GerochSplitting} and refined by Bernal \& S\'anchez \cite{BS04,BS04a} is that they admit global smooth splittings into the time and space parts, and each such splitting arises from a \emph{Cauchy temporal function}, i.e. a temporal function whose all level sets are Cauchy hypersurafces.

\begin{thm}(Geroch--Bernal--S\'anchez)
\label{GBS}
Spacetime $(\M,g)$ is globally hyperbolic iff it admits Cau\-chy temporal functions. For any such function $\T$ there exists an isometry $\Xi: \M \rightarrow \sR \times \itSigma$, called the Geroch--Bernal--S\'anchez (GBS) splitting, such that $\itSigma \vc \T^{-1}(0)$ and $\T = \pi^0 \circ \Xi$ with $\pi^0$ being the canonical projection on the $\sR$ component. Moreover, the metric can be expressed as
\begin{align}
\label{GBSmetric}
g = -\alpha \, d\T \otimes d\T + \bar{g},
\end{align}
where $\alpha \vc |\nabla \T|^{-2} = -g(\nabla \T,\nabla \T)^{-1}$ is a positive smooth function and $\bar{g}$ is a 2-covariant symmetric tensor field on $\M$ whose restriction to the Cauchy hypersurface $\M_t \vc \T^{-1}(t)$ is a Riemannian metric for every $t \in \sR$ and whose radical at each $p\in\M$ is spanned by $\nabla \T_p$.
\end{thm} 

Let us then fix a Cauchy temporal function $\T$ furnishing the `global time'. Consequently, the smooth spacelike Cauchy hypersurfaces $\M_t$ become the time slices and can be interpreted as the (curved) `hypersurfaces of simultaneity'\footnote{Alternatively, one might regard $\T$ as giving rise to a \emph{synchronizable observer field} $U = -\tfrac{\nabla \T}{|\nabla \T|}$ (cf. \cite[pp. 358--360]{BN83}), whose restspaces are all assumed Cauchy.}. Suppose we want to describe the time-evolution, over a time interval $I$, of some nonnegative physical quantity $Q$ distributed in space --- be it mass, energy or probability --- the total amount of which is conserved. To this end, after normalizing the total amount of $Q$ to one, we introduce the following notion.
\begin{defn}
A \emph{$\T$-evolution} (of a probability measure) is a map $\mu: I \rightarrow \Pf(\M_I)$, where $\M_I \vc \T^{-1}(I)$, such that each $\mu(t) \cv \mu_t$ lives on the corresponding time slice, i.e. $\supp \mu_t \subset \M_t$. 
In addition, we say that $\mu$ is \emph{causal} if it is $\preceq$-monotone, that is if
\begin{align}
\label{causal_measure_map0}
\forall \, s,t \in I \quad s \leq t \ \Rightarrow \ \mu_s \preceq \mu_t.
\end{align}
\end{defn}

Of course, in accordance with the tenets of general relativity, there is no canonical choice of a Cauchy temporal function and a GBS splitting. Different physicists (or `global observers') might use different splittings in their descriptions of dynamical phenomena. The crucial thing, however, is that no description should be fundamentally privileged. They must be `covariant', in the sense that there exist transformation rules allowing one to unambiguously pass from one description to another. This, in turn, is guaranteed by the existence of `invariant' geometrical objects, which do not depend on the choice of the GBS splitting. The standard example of such an invariant is a worldline, i.e. an unparametrized causal curve in spacetime. The choice of a splitting specifies a parametrization of the curve, enabling the interpretation in terms of a moving pointlike particle.

The aim of this paper is to find invariant, splitting-independent objects behind the notion of a causal $\T$-evolution, along with the transformation rules between descriptions involving different Cauchy temporal functions. 

To this end, we investigate the relationship between causally evolving probability measures and the low-regular version of \emph{continuity equation}, which we introduce below.

In the smooth setting, the continuity equation is simply $\dyw J = 0$, where $J$ is the 4-current describing the flux of $Q$. In the case of Minkowski spacetime $\sM$, by expressing $J$ in a chosen inertial frame as $(\rho, \rho \bv)$ with $\bv$ subluminal ($|\bv|^2 \leq 1$), the equation takes the form $\partial_t \rho + \nabla_\bx \cdot (\rho \bv) = 0$ known from elementary physics. This can be naturally extended beyond the smooth setting, wherein $\rho$ is replaced by a $\pi^0$-evolution $\mu: I \rightarrow \Pf(\sM_I)$ (with $\sM_I = I \times \sR^3$) and the equation is now understood in the distributional sense. Concretely\footnote{Cf. \cite[Definition 1.4.1]{Crippa2007PhD}, where $I = [0,T]$, whereas $\mu_t$'s are signed measures and, technically, they are regarded as living not on $\sM_I$, but rather on $\sR^d$ (for any $d$).}, we say that $I \ni t \mapsto \mu_t$ satisfies the continuity equation if
\begin{align}
\label{conteq_Minkowski}
\forall \Phi \in C^\infty_\cc(\sM_I^\circ) \qquad \int_I \int_{\sM_I} \left( \partial_t \Phi + \bv_t \cdot \nabla_\bx \Phi \right) d\mu_t dt = 0
\end{align}
for some Borel velocity field $\bv: \sM_I^\circ \rightarrow \sR^3$, $(t,\bx) \mapsto \bv_t(\bx)$, where $\sM_I^\circ = I^\circ \times \sR^3$ denotes the interior\footnote{See Appendix \hyperref[sec::AppA]{A} for more on the topological notation used.} of $\sM_I$. This distributional version of continuity equation (with $\bv$ not necessarily subluminal) has been extensively studied within the optimal transport theory and metric geometry \cite{ambrosio2008gradient,ambrosio2008continuity,AB08,BB2000,Bernard12}.

Moving now to the more general case of a globally hyperbolic spacetime $\M$, the analogue of (\ref{conteq_Minkowski}) can be derived as follows. Consider first a \emph{smooth} 4-current $\rho Y$, where $\rho \in C^\infty(\M_I^\circ)$ and $Y$ is a smooth 4-velocity field, which is by definition future-directed causal and normalized in accordance with the GBS splitting used, i.e. such that $g(Y, -\tfrac{\nabla \T}{|\nabla \T|}) = -1$ or, equivalently,
\begin{align}
\label{conteq-1}
Y\T = |\nabla \T|.
\end{align}
Notice that $Y$ can be expressed as
\begin{align*}
Y = - \tfrac{\nabla \T}{|\nabla \T|} + V,
\end{align*}
where $V$ is a (smooth) spacelike vector field satisfying $V\T = 0$, interpreted by the observer as the 3-velocity, and the causality of $Y$ amounts to the former's subluminality, $|V| = g(V,V)^{1/2} \leq 1$. 

The continuity equation reads
\begin{align}
\label{conteq0}
\dyw \rho Y = 0.
\end{align}
In order to pass to the distributional formulation, we integrate (\ref{conteq0}) against every test function analogously as in (\ref{conteq_Minkowski}) to obtain the equivalent condition\footnote{As a side remark, notice that $\M_I^\circ = \T^{-1}(I)^\circ = \T^{-1}(I^\circ)$, where the last equality holds because $\T$, being a composition of open maps $\pi^0 \circ \Xi$, is itself an open map and so it commutes with the operation of passing to the interior \cite[1.4.C]{Engelking}.}
\begin{align}
\label{conteq1}
\forall \Phi \in C_\cc^\infty(\M_I^\circ) \quad \int_{\M_I} \dyw(\rho Y) \Phi \, d\vol_g = 0,
\end{align}
where $\vol_g$ is the volume measure associated to $g$. Using the divergence theorem \cite[Lemma 10.8]{HR09}, we can rewrite (\ref{conteq1}) as
\begin{align}
\label{conteq2}
\forall \Phi \in C_\cc^\infty(\M_I^\circ) \quad \int_{\M_I} Y(\Phi) \rho \, d\vol_g = 0.
\end{align}
Notice now that by (\ref{GBSmetric}) we have $\vol_g = \int_I |\nabla \T|^{-1} \vol_{\bar{g}_t} \, dt$, where $\vol_{\bar{g}_t}$ denotes the volume measure supported on $\M_t$, associated to the Riemannian metric $\bar{g}|_{\M_t}$. Thus, the continuity equation becomes
\begin{align}
\label{conteq3}
\forall \Phi \in C_\cc^\infty(\M_I^\circ) \quad \int_I \int_{\M_I} |\nabla \T|^{-1} Y(\Phi) \rho \, d\vol_{\bar{g}_t} \, dt = 0.
\end{align}
Finally, in order to obtain the analogue of (\ref{conteq_Minkowski}) for a $\T$-evolution, we replace $\rho \, d\vol_{\bar{g}_t}$ with $d\mu_t$ and introduce $X \vc |\nabla \T|^{-1} Y$, which is a smooth future-directed causal vector field satisfying $X\T = 1$ by (\ref{conteq-1}). Now we are in the position to lower the regularity of $X$, and we arrive at the following definition.
\begin{defn}
\label{def_def_conteq}
For any $\T$-evolution of a probability measure $\mu: I \rightarrow \Pf(\M_I)$ introduce $\eta \vc \int_I \mu_t dt$, which is a Radon measure on $\M_I$. We say that $\mu$ satisfies the continuity equation (in the distributional sense) if
\begin{align}
\label{def_conteq}
\forall \Phi \in C^\infty_{\cc}(\M_I^\circ) \quad \quad \int_{\M_I} X\Phi \, d\eta = 0
\end{align}
for some Borel vector field $X$ which is $\eta$-a.e. future-directed causal and such that $X\T = 1$.
\end{defn}
\begin{rem}
\label{Linfrem}
Notice that $X$ appearing in the above definition is automatically an $L^\infty_\loc(\eta)$-vector field --- i.e. its local coordinate functions are locally $\eta$-essentially bounded (see Definition \ref{def_VF} below for details). Indeed, since it satisfies $X\T = 1$ $\eta$-a.e., it can be expressed as
\begin{align}
\label{Linfrem1}
X = - \tfrac{\nabla \T}{|\nabla \T|^2} + \tfrac{V}{|\nabla \T|},
\end{align}
where $V \vc |\nabla \T|X + \tfrac{\nabla \T}{|\nabla \T|}$ is a Borel vector field satisfying
\begin{align}
\label{Linfrem2}
V\T = 0 \quad \textnormal{and} \quad \bar{g}(V,V) = g(V,V) \leq 1 \quad \eta\textnormal{-a.e.}.
\end{align}
Observe that for any $\Psi \in C^\infty(\M)$ the map $V\Psi$ is locally $\eta$-essentially bounded, because 
\begin{align*}
|V\Psi| = |g(V,\nabla\Psi)| = |\bar{g}(V,\nabla\Psi)| \leq \bar{g}(V,V)^{1/2}\bar{g}(\nabla\Psi,\nabla\Psi)^{1/2} \leq \bar{g}(\nabla\Psi,\nabla\Psi)^{1/2}
\end{align*}
$\eta$-a.e., and the rightmost function is of course locally bounded. But this means that $V$ (and hence $X$) is an $L^\infty_\loc(\eta)$-vector field, because for $\Psi$ one can substitute the (suitably extended) coordinate maps.
\end{rem}


Intuitively, the vector field $X$ describes the `probability flux' as seen by the observer. Since $X$ is explicitly required to be causal, one might expect that the probability measure --- and the physical quantity $Q$ it models --- evolves \emph{causally}, i.e. that $\mu$ satisfies (\ref{causal_measure_map0}). This has been indeed shown in the case of Minkowski spacetime \cite{PRA2017}. Two immediate questions arise: Does this remain true in \emph{any} globally hyperbolic spacetime? Does the \emph{converse} implication hold as well? Here we show that the answer to both these questions is positive. In other words, we demonstrate that the low-regular continuity equation (\ref{conteq_Minkowski}) in fact \emph{characterizes} the causal $\T$-evolutions of probability measures.

In order to achieve this goal, we build upon the results of \cite{Miller17a}, where another characterization of causal $\T$-evolutions of measures has been established. It involves spaces $A_\T^I$ of continuous causal curves $\gamma: I \rightarrow \M$, parame\-trized in such a way that $\T(\gamma(t)) = t$ for all $t \in I$, endowed with the compact-open topology induced from $C(I,\M)$ (cf. Definition \ref{AIT_def} below). When topologized in this way, $A_\T^I$ becomes a locally compact Polish (i.e. separable and completely metrizable) space, and with the help of an apt compactness argument we find a way to `encapsulate' the given causal $\T$-evolution $\mu$ into a single probability measure $\sigma$ on $A_\T^I$. Crucially --- and somewhat surprisingly --- the said topology on $A_\T^I$ `feels' also the \emph{differential} properties of the curves, and we exploit this fact to construct a vector field $X$ in a sense `tangent to $\sigma$', with which $\mu$ is then shown to satisfy (\ref{def_conteq}).

Summing up, the main result of the paper is as follows.
\begin{thm}
\label{premain}
For a $\T$-evolution $\mu: I \rightarrow \Pf(\M_I)$ the following conditions are equivalent:
\begin{enumerate}[(\itshape i)]
\item $\mu$ is \emph{causal}, i.e. it satisfies (\ref{causal_measure_map0}).
\item There exists $\sigma \in \Pf(A_\T^I)$ such that $(\ev_t)_\sharp \sigma = \mu_t$ for every $t \in I$, where $\ev_t: A_\T^I \rightarrow \M$, $\ev_t(\gamma) \vc \gamma(t)$ denotes the evaluation map.
\item If we denote $\eta \vc \int_I \mu_t dt$, there exists an $L^{\infty}_\loc(\eta)$-vector field $X$, future-directed causal and satisfying $X\T = 1$ $\eta$-a.e., with which $\mu$ satisfies the continuity equation (\ref{def_conteq}).

What is more, for any $\Phi \in C^\infty_\cc(\M_I^\circ)$ the maps $\Lambda_\mu(\Phi)(t) \vc \int_{\M_I} \Phi \, d\mu_t$ and $\Lambda_\mu(X\Phi)(t) \vc \int_{\M_I} X \Phi \, d\mu_t$ are well defined for $t \in I$ a.e. and we have that $\Lambda_\mu(\Phi) \in W^{1,\infty}(I)$ along with $\Lambda_\mu(\Phi)' = \Lambda_\mu(X\Phi)$.
\end{enumerate}
\end{thm}

Several remarks are in order. Let us begin with a simple, single-curve illustration.

\begin{ex}
\label{ex1}
Consider the $\T$-evolution of a Dirac measure $\mu_t \vc \delta_{\gamma(t)}$ with $\gamma \in A_\T^I$ being a $C^1$ causal curve. Of course, this evolution is causal (just compare (\ref{causal_measure_map0}) with (\ref{AIT_def1}) below) and the only measure $\sigma \in \Pf(A_\T^I)$ satisfying $(\ev_t)_\sharp \sigma = \delta_{\gamma(t)}$ is, trivially, $\sigma \vc \delta_\gamma$. As for the $L^{\infty}_\loc(\eta)$-vector field $X$, in this case it can be defined via
\begin{align}
\label{premain1}
\forall \Psi \in C^\infty(\M_I^\circ) \qquad X\Psi(q) \vc \left\{ \begin{array}{ll} 
					 (\Psi \circ \gamma)'(t) & \textnormal{for } q = \gamma(t)
					 \\
					 0 & \textnormal{otherwise}.
\end{array} \right.
\end{align}
Notice that, since the measure $\eta \vc \int_I \delta_{\gamma(t)} dt$ is supported on $\gamma(I)$, the values $X\Psi(q)$ for $q \not\in \gamma(I)$ are irrelevant. Thus defined $X$ clearly satisfies $X\T = 1$ $\eta$-a.e., whereas the continuity equation is satisfied because, for any $\Phi \in C^\infty_{\cc}(\M_I^\circ)$,
\begin{align*}
\int_{\M_I} X\Phi \, d\eta = \int_I (\Phi \circ \gamma)'(t) dt = \int_{\gamma^{-1}(\supp \Phi)} (\Phi \circ \gamma)'(t) dt = 0,
\end{align*}
since one can always choose $[a,b] \subset I^\circ$ so that $\gamma^{-1}(\supp \Phi) \subset \T(\supp \Phi) \subset (a,b)$, what indeed yields $\int_I (\Phi \circ \gamma)'(t) dt = \Phi(\gamma(b)) - \Phi(\gamma(a)) = 0$.

The $L^\infty_\loc(\eta)$-regularity of $X$ follows now from Remark \ref{Linfrem1}. Finally, the equality $\Lambda_\mu(\Phi)' = \Lambda_\mu(X\Phi)$ appearing in (iii) in this case reduces to $(\Phi \circ \gamma)' = X\Phi \circ \gamma$, which follows trivially from (\ref{premain1}).
\end{ex}

What Theorem \ref{premain} says is that the above threefold description of a causal $\T$-evolution of a Dirac measure generalizes to any (Borel) probability measure.

Equivalence (i) $\Leftrightarrow$ (ii), as already mentioned, has been established in \cite{Miller17a}. However, the spaces of causal curves have been defined there somewhat broadlier (the curves $\gamma$ being parametrized according to a rule less restrictive than $\T(\gamma(t)) = t$) and the proof can be greatly simplified by exploiting the local compactness of $A_\T^I$. With all that in mind, in the current paper we present a new compact proof of (i) $\Leftrightarrow$ (ii).

In the proof of (ii) $\Rightarrow$ (iii), the announced construction of a vector field $X$ `tangent to $\sigma$' relies, loosely speaking, on `integrating (\ref{premain1}) with respect to $\sigma$', what is made precise by formula (\ref{def_X0}) below. One might naturally ask whether \emph{every} vector field $X$ satisfying (iii) arises in this matter, i.e. whether there always exists $\sigma \in \Pf(A_\T^I)$, to which $X$ is `tangent'. The answer is known to be positive for Minkowski spacetime, what is a direct consequence of the so-called \emph{superposition principle}\footnote{Not to be confused with the superposition principle in linear systems theory!}, which is a result on the verge of the PDE theory and optimal transport theory (cf. \cite[Theorem 3]{Bernard12}, \cite[Theorem 3.2]{ambrosio2008continuity} or \cite[Theorem 6.2.2]{Crippa2007PhD} for various formulations of the principle and \cite[Theorem 3]{PRA2017} for its application to causality in Minkowski spacetime). The question whether analogous result holds for every globally hyperbolic spacetime, however, lies beyond the scope of the present work and will be addressed in the follow-up paper \cite{MillerSuperposition}. Here, instead, we complete the chain of equivalences by showing that (iii) $\Rightarrow$ (i).

Let us emphasize that, in general, neither $\sigma$ nor $X$ are uniquely determined by a given causal $\T$-evolution $\mu$, as the following example shows.

\begin{ex}
Consider the cylindrical spacetime $\M \vc \sR \times S^1$ with the flat metric $g \vc -dt^2 + d\vartheta^2$ and the canonical projection $\pi^0$ as the Cauchy temporal function. Let also $\mu_t \vc \delta_t \times \lambda$ for any $t \in \sR$, where $\lambda$ is the normalized Haar measure on $S^1$. Then in fact $\eta = \cL \times \lambda$, where $\cL$ is the Lebesgue measure on $\sR$, and the smooth vector field $X \vc \partial_t + a \partial_\vartheta$ satisfies (iii) for any fixed $|a| \leq 1$. Moreover, for any $\vartheta \in S^1$ the curve $\gamma_{\vartheta}(t) \vc (t, \vartheta + at)$ is clearly an integral curve of $X$ and it is not difficult to convince oneself that the map $\itGamma: S^1 \rightarrow A_{\pi^0}^\sR$, $\itGamma(\vartheta) \vc \gamma_\vartheta$ is well-defined and continuous. The pushforward measure $\sigma \vc \itGamma_\sharp \lambda$ satisfies (ii), which can be easiest seen by integrating any $f \in C_\cb(\M)$ with respect to $(\ev_t)_\sharp \sigma$ for any $t \in \sR$. Indeed, one has that
\begin{align*}
& \int_\M f d(\ev_t)_\sharp \sigma = \int_\M f d(\ev_t \circ \itGamma)_\sharp \lambda = \int_{S^1} f(t, \vartheta + at) d\lambda(\vartheta) = \int_{S^1} f(t, \vartheta) d\lambda(\vartheta) = \int_\M f d\mu_t,
\end{align*}
where we have used the translational invariance of $\lambda$. To summarize, we have found a \emph{continuum} of measures $\sigma$ and vector fields $X$ (parametrized by $a \in [-1,1]$), all leading to the `constant' evolution $\mu_t = \delta_t \times \lambda$. Although this is in stark contrast with the single-worldline scenario above, it is not very surprising --- intuitively, both $X$ and $\sigma$ contain the information about the motion of the `infinitesimal portions' of probability, and that information might be lost when passing to the `bare' collection of instantaneous measures $\{\mu_t\}_t$ with no relation between them (other than the causal precedence). In other words, each of descriptions (ii) and (iii) gives the fuller kinematical picture of the $\T$-evolution than (i) alone.
\end{ex}

Although Theorem \ref{premain} manifestly employs the notion of a `global time', it gives us the means to find the desired invariant, splitting-independent objects behind the notion of a causal $\T$-evolution along with the transformation rules for passing between different `global observers'. More concretely, we demonstrate that, if $\eta_\bA$, $X_\bA$ and $\eta_\bB$, $X_\bB$ are the measures and vector fields appearing in description (iii) obtained through the use of Cauchy temporal functions $\T_\bA$ and $\T_\bB$, respectively, then
\begin{align*}
\eta_\bB = X_\bA \T_\bB \, \eta_\bA, \qquad \textrm{whereas} \qquad X_\bB = \tfrac{1}{X_\bA \T_\bB} X_\bA,
\end{align*}
and so their product (suitably understood) is the desired invariant, $\eta_\bA X_\bA = \eta_\bB X_\bB$. This is in fact quite expected in light of $\eta X$ being the low-regular counterpart of $\vol_g J$ appearing in the derivation of Definition \ref{def_def_conteq} above. 

Another invariant object is obtained from description (ii), because the measure $\sigma \in \Pf(A_\T^\sR)$ and the space $A_\T^\sR$ itself can be `deparametrized', what leaves is with a well defined probability measure $\upsilon$ on the space of worldlines. In other words, just like a single worldline is an observer-independent, geometrical object modeling the time-evolution of a pointlike particle, so does a probability measure on the space of worldlines in case of a time-evolving nonlocal entity (see \cite[Sec. 2]{Miller17a} for a more detailed discussion).
\\

The plan of the paper is as follows. In Sec. \ref{sec::Polish_spaces} we (re)introduce Polish spaces of causal curves $A_\T^I$ endowed with the compact-open topology. We do so not only to make the paper self-contained, but, more importantly, to unravel the announced sensitivity of the employed topology to the curves' differential properties. This is realized first by introducing the ``weak-$H^1_\loc$-topology'' on $A_\T^I$, obtained through embedding the latter into a suitable $H^1_\loc$-type Sobolev space, and then by carefully proving that the weak topology pulled back on $A_\T^I$ from that Sobolev space is in fact \emph{equal} to the compact-open topology. In Sec. \ref{sec::vector_fields} we look more closely into the notion of a $L^p_\loc$-vector field. Most importantly, we provide its alternative, equivalent definition that is much better suited for proving Theorem \ref{premain}. The latter is done in Sec. \ref{sec::proof} with the necessary tools and lemmas introduced or recalled when needed. Finally, Sec. \ref{sec::discussion} discusses how the descriptions (i)--(iii) transform between various observers and unveils the generally covariant objects behind them. The article is supplemented by three appendices. Appendix \hyperref[sec::AppA]{A} contains several notational remarks. Appendix \hyperref[sec::AppB]{B} summarizes the definitions and facts from Lorentzian causality theory needed in the current investigations. Finally, Appendix \hyperref[sec::AppC]{C} provides a detailed exposition of the theory of Sobolev spaces of univariate functions, along with some technical lemmas used in the proofs throughout the paper.


\section{Polish spaces of continuous causal curves}
\label{sec::Polish_spaces}

Let $(\M,g)$ be a globally hyperbolic spacetime and let $\T$ be a fixed Cauchy temporal function. We shall be considering the following spaces of continuous future-directed causal curves parametrized in accordance with $\T$.
\begin{defn}
\label{AIT_def}
For any interval $I \subset \sR$ we define $A_\T^I$ to be the space of continuous curves $\gamma: I \rightarrow \M$ such that $\T \circ \gamma = \id_I$ and
\begin{align}
\label{AIT_def1}
\forall \, s,t \in I \quad s \leq t \ \Rightarrow \ \gamma(s) \preceq \gamma(t),
\end{align}
endowed with the compact-open topology induced from $C(I,\M)$.
\end{defn} 

Condition (\ref{AIT_def1}) extends the notion of a (future-directed) causal curve beyond the piecewise $C^1$ ones (cf. Appendix \hyperref[sec::AppB]{B}). When itself extended onto $\T$-evolutions of probability measures, it gives rise to condition (\ref{causal_measure_map0}) above.

The compact-open topology on $C(I,\M)$ is exactly the topology of $h$-uniform convergence on compact subsets, where $h$ is any complete Riemannian metric on $\M$, which in particular means that $C(I,\M)$ is a Polish space \cite[Example A.10]{KerrLi}.
\begin{prop}
\label{AIT_Polish}
$A_\T^I$ is a closed subspace of $C(I,\M)$ and hence a Polish space itself.
\end{prop}
\begin{proof}
Let $(\gamma_n) \subset A_\T^I$ converge $h$-uniformly on compact subsets to $\gamma \in C(I,\M)$. Since this implies pointwise convergence, one can write that, for any $t \in I$,
\begin{align*}
\T(\gamma(t)) = \lim_{n \rightarrow +\infty} \T(\gamma_n(t)) = t
\end{align*}
along with, for any $s,t \in I$, $s \leq t$,
\begin{align*}
\gamma(s) = \lim_{n \rightarrow +\infty} \gamma_n(s) \preceq \lim_{n \rightarrow +\infty} \gamma_n(t) = \gamma(t),
\end{align*}
where we have used the topological closedness of $J^+$ (Proposition \ref{PropCS}).
\end{proof}

The compact-open topology is by no means the only one the sets of causal curves can be endowed with. Standard choices include the so-called $C^0$-topology \cite{GerochSplitting,HawkingEllis,Penrose1972} defined for \emph{compact} causal curves (i.e. those with both endpoints) and the compact-open topology on the space of curves param\-e\-trized by their arc-length \cite{YCB67,LimitCurves} (see also \cite{Samann16,SanchezProgress} for more details). In fact, one can show that the sequence $(\gamma_n) \subset A_\T^{[a,b]}$ converges in the compact-open topology iff it converges in the $C^0$-topology \cite[Theorem 3]{Miller17a}. One can thus regard the former as a way of extending the latter to the case of noncompact causal curves --- at least when a (not necessarily Cauchy) temporal function is available.

However, the compact-open topology seems to be too weak for our current purposes. Causal curves, when suitably parametrized, are well known to be locally Lipschitz-continuous (cf. \cite[2.26]{Penrose1972}, \cite[Remark 3.18]{MS08} or Proposition \ref{Liploc} below) and hence a.e. differentiable --- a property not taken into account by the uniform convergence. On the other hand, as already signalized in Sec. \ref{sec::intro}, proving (iii) in Theorem \ref{premain} evidently requires tools sensitive to the curves' differential properties.

Motivated by these observations, we shall endow $A_\T^I$ with the weak topology pulled back from the local Sobolev space $H^1_\loc(I,\sR^N)$ and study its properties. The very idea of studying Sobolev-type regularity of causal curves is by no means new, see e.g. \cite{Candela}, where the Sobolev space structure was defined in terms of local charts (drawing from \cite{Palais}). However, although this approach works well for the strong topology, it becomes cumbersome for the weak one. Here we propose another approach which, instead of local charts, employs a `nice' global embedding provided by the following variant of the celebrated Nash theorem \cite{Nash56}, considered and proven by M\"uller \cite{NashMuller}.
\begin{thm}(Nash--M\"uller)
\label{NM}
If $(\M,h)$ is a complete Riemannian manifold, then there exists $N \in \sN$ and a closed isometric $C^\infty$-embedding $i: \M \hookrightarrow \sR^N$, where $(\sR^N, \cdot)$ is the $N$-dimensional Euclidean space.
\end{thm}

Any such an embedding $i$ will be called a \emph{Nash--M\"uller embedding}. M\"uller even specifies that $N = \tfrac{1}{2} m (m+1)(3m + 11) + 1$, where $m \vc \dim \M$, i.e. one more than in the original result by Nash, who neither assumed the Riemannian metric to be complete nor required the embedding to be closed.

In applying the above theorem, it will be very convenient to fix a concrete Riemannian metric on $\M$ associated to the spacetime metric (\ref{GBSmetric}). Concretely, consider the Riemannian metric $h_0 \vc \alpha \, d\T \otimes d\T + \bar{g}$. By the Nomizu--Ozeki theorem \cite{NomizuOzeki}, it is conformally equivalent to a \emph{complete} Riemannian metric. In other words, there exists a smooth positive function $\theta \in C^\infty(\M)$ such that
\begin{align}
\label{R_metric}
h \vc \theta \alpha \, d\T \otimes d\T + \theta \bar{g} = 2 \theta \alpha \, d\T \otimes d\T + \theta g
\end{align}
is a complete metric. From now on, it will be our Riemannian metric of choice on $\M$.

In terms of $h$ and its associated distance function $d_h$ it is not hard to demonstrate the local Lipschitz-continuity of the causal curves parametrized in accordance with $\T$.
\begin{prop}
\label{Liploc}
Every $\gamma \in A_\T^I$ is locally Lipschitz-continuous.
\end{prop}
\begin{proof}
Fix any $(a,b) \Subset I$ and pick any $t_1,t_2 \in (a,b)$, $t_1 < t_2$. By assumption $\gamma(t_1) \preceq \gamma(t_2)$, i.e., there exists a piecewise-smooth future-directed causal curve $\rho: [0,1] \rightarrow \M$ such that $\rho(0) = \gamma(s)$ and $\rho(1) = \gamma(t)$. We claim that $\rho$ can be reparametrized so as to become a piecewise-smooth element of $A_\T^{[t_1,t_2]}$.

Indeed, begin by noticing that $\T \circ \rho$ is a piecewise-smooth map satisfying $(\T \circ \rho)'(t) = \rho'(t)(\T) > 0$ whenever $\rho'(t)$ exists (cf. Lemma \ref{lem_CT}). This in particular means that it is strictly increasing and maps the interval $[0,1]$ onto $[t_1,t_2]$. All this implies, in turn, that the inverse map $(\T \circ \rho)^{-1}: [t_1,t_2] \rightarrow [0,1]$ is well defined and piecewise-smooth by the inverse function theorem, with its derivative existing at every $\tau = (\T \circ \rho)(t)$ for which $\rho'(t)$ exists and given by $[(\T \circ \rho)^{-1}]'(\tau) = \tfrac{1}{(\T \circ \rho)'(t)} > 0$.

We thus conclude that $\hat{\rho} \vc \rho \circ (\T \circ \rho)^{-1}$ is a piecewise-smooth future-directed causal curve that clearly satisfies $\T \circ \hat{\rho} = \id_{[t_1,t_2]}$ and connects $\gamma(t_1)$ with $\gamma(t_2)$. All this allows us to write
\begin{align}
\nonumber
& d_h(\gamma(t_2),\gamma(t_1)) \leq \int_{t_1}^{t_2} \sqrt{h(\hat{\rho}',\hat{\rho}')}\, d\tau = \int_{t_1}^{t_2} \sqrt{\theta(\hat{\rho})} \sqrt{2 \alpha(\hat{\rho})\smash[b]{\underbrace{|d\T(\hat{\rho}')|^2}_{=\,1} + \underbrace{g(\hat{\rho}',\hat{\rho}')}_{\leq \, 0}}} \, d\tau
\\[4pt]
\label{locLip1}
& \leq \int_{t_1}^{t_2} \sqrt{2 \theta(\hat{\rho}) \alpha(\hat{\rho})} \, d\tau \leq |t_2-t_1| \max\nolimits_{q \in J^+(\gamma(a)) \cap J^-(\gamma(b))} \sqrt{2 \theta(q) \alpha(q)},
\end{align}
where we have suppressed the argument $\tau$ under the integrals. Notice that we rely here on the global hyperbolicity of $\M$, as the compactness of $J^+(\gamma(a)) \cap J^-(\gamma(b))$ ensures the existence of the maximum.
\end{proof}

The above proposition has three important corollaries.
\begin{cor}
\label{diffae}
For any $\gamma \in A_\T^I$ the tangent vector $\gamma'(t)$ exists a.e. and is future-directed causal.
\end{cor}
\begin{proof}
Diffrentiability a.e. follows from the local Lipschitz-continuity of $\gamma$ by the Rademacher theorem. To show that $\gamma'(t)$ is future-directed causal (whenever it exists) we employ Lemma \ref{lem_CT}. Indeed, for any smooth causal function $f$ one has $\gamma'(t)(f) = (f \circ \gamma)'(t) \geq 0$ by the very definition of causal function.
\end{proof}
\begin{cor}
\label{locLip}
Let $i: \M \hookrightarrow \sR^N$ be a Nash--M\"uller embedding. Then the map $i_\ast: A^I_\T \rightarrow H^1_\loc(I,\sR^N)$ defined as $i_\ast(\gamma) \vc i \circ \gamma$ is a well-defined injection.
\end{cor}
\begin{proof}
Since every $\gamma \in A_\T^I$ is locally Lipschitz-continuous, so is $i \circ \gamma$. The latter is equivalent to saying that $i_\ast(\gamma) \in W^{1,\infty}_\loc(I,\sR^N)$ (cf. \cite[Section 5.8, Theorem 4]{Evans}), which by (\ref{Holder}) implies that $i_\ast(\gamma) \in H^1_\loc(I,\sR^N)$. The injectivity of $i_\ast$ follows immediately from the injectivity of $i$.
\end{proof}

Whenever $\gamma'(t)$ exists, there also exists the strong derivative $(i \circ \gamma)'(t) = di_{\gamma(t)}\gamma'(t)$. Since the former exists a.e., the latter can be regarded as the weak derivative of $i_\ast(\gamma)$. Moreover, by the isometricity of $i$ and by Corollary \ref{diffae} we can write that
\begin{align}
\label{derbound}
& |i_\ast(\gamma)'(t)| = |di_{\gamma(t)}\gamma'(t)| = \sqrt{h(\gamma'(t),\gamma'(t))} \leq \sqrt{2 \theta(\gamma(t))\alpha(\gamma(t))} \quad \textnormal{for } t \in I \textnormal{ a.e.}
\end{align}
Denoting the strong and weak derivatives with the same symbol $'$ should not lead to any confusion.

\begin{cor}
\label{C0viaNM}
Let $i: \M \hookrightarrow \sR^N$ be a Nash--M\"uller embedding. Then $i_\ast(A_\T^I)$ is a closed subset of $C(I,\sR^N)$ endowed with the compact-open topology.
\end{cor}
\begin{proof}
Suppose $(i \circ \gamma_n)$ converges to some $u \in C(I,\sR^N)$ uniformly on compact sets. Since $i$ is by definition a closed embedding, we must have $u = i \circ \gamma$ for some $\gamma \in \M^I$. But since $i$ is also an isometry, we obtain that $\gamma_n \rightarrow \gamma$ uniformly on compact sets, and so $\gamma \in A_\T^I$ by Proposition \ref{AIT_Polish}.
\end{proof}

We are finally ready to define the announced topology on $A_\T^I$.

\begin{defn}
\label{weakH1loctop}
Let $i: \M \hookrightarrow \sR^N$ be a Nash--M\"uller embedding. By the \emph{weak-$H^1_\loc$-topology} on $A^I_\T$ we shall understand the coarsest topology under which $i_\ast$ is weakly continuous. In other words, a net $(\gamma_\lambda) \subset A^I_\T$ converges to $\gamma$ in this topology (in symbols $\gamma_\lambda \xrightharpoonup{i} \gamma$) if $i_\ast(\gamma_\lambda) \rightharpoonup i_\ast(\gamma)$ in $H^1_\loc(I,\sR^N)$ or, equivalently, if $\langle i \circ \gamma_\lambda, v \rangle_{H^1[a,b]} \rightarrow \langle i \circ \gamma , v \rangle_{H^1[a,b]}$ for every $[a,b] \subset I^\circ$ and $v \in H^1([a,b],\sR^N)$.
\end{defn}
\begin{rem}
\label{rem_weakH1loctop}
Thanks to its injectivity (Corollary \ref{locLip}), $i_\ast: A^I_\T \hookrightarrow H^1_\loc(I,\sR^N)$ becomes a topological embedding. Observe also that $\gamma_\lambda \xrightharpoonup{i} \gamma$ in $A^I_\T$ iff $\gamma_\lambda|_{[a,b]} \xrightharpoonup{i} \gamma|_{[a,b]}$ in $A^{[a,b]}_\T$ for every $[a,b] \subset I^\circ$.
\end{rem}

Clearly, nothing prohibits us from analogously defining the \emph{strong-$H^1_\loc$-topology} on $A_\T^I$, but here we shall focus exclusively on the weak-$H^1_\loc$-topology, eventually unraveling that it is actually \emph{equal} to the compact-open topology. Note that this in particular means that it does not depend on the specific choice of the Nash--M\"uller embedding. 

First, however, we need a couple of lemmas.
\begin{lem}
\label{lem_ev_contA}
For any $t \in I^\circ$ the evaluation map $\ev_t$ is weakly-$H^1_\loc$-continuous.
\end{lem}
\begin{proof}
By Lemma \ref{lem_ev_cont} and the continuity of $i_\ast: A^I_\T \rightarrow H^1_\loc(I,\sR^N)$, we have that the map $A^I_\T \ni \gamma \mapsto i(\gamma(t)) \in i(\M)$ is weakly-$H^1_\loc$-continuous. The fact that $i(\M)$ is homeomorphic with $\M$ finishes the proof.
\end{proof}

For the next results we need to consider certain specific subsets of causal curves. Namely, for any subsets $\U,\W \subset \M$ let us define
\begin{align*}
A^{[a,b]}_\T(\U,\W) \vc A^{[a,b]}_\T \cap \ev_a^{-1}(\U) \cap \ev_b^{-1}(\W) = \{ \gamma \in A^{[a,b]}_\T \ | \ \gamma(a) \in \U, \gamma(b) \in \W \}.
\end{align*}
\begin{lem}
\label{lem_compact}
For any compact subsets $\K_1,\K_2 \subset \M$ the set $A^{[a,b]}_\T(\K_1,\K_2)$ is metrizable and compact in the weak-$H^1_\loc$-topology.
\end{lem}
\begin{proof}
Firstly, observe that $i_\ast(A^{[a,b]}_\T(\K_1,\K_2))$ is $\|.\|_{H^1[a,b]}$-bounded. Indeed, since the set $J^+(\K_1) \cap J^-(\K_2)$ is compact (Proposition \ref{PropCCC}), we have that, for any $\gamma \in A^{[a,b]}_\T(\K_1,\K_2)$,
\begin{align*}
\int_a^b |i(\gamma(t))|^2 dt & \leq (b-a) \max\nolimits_{q \in J^+(\K_1) \cap J^-(\K_2)} |i(q)|^2 < \infty
\\
\textnormal{and} \quad \int_a^b |i(\gamma(t))'|^2 dt & \leq 2 (b-a) \max\nolimits_{q \in J^+(\K_1) \cap J^-(\K_2)} \theta(q)\alpha(q) < \infty,
\end{align*}
where in the second inequality we have used (\ref{derbound}).

We have thus proven that $i_\ast(A^{[a,b]}_\T(\K_1,\K_2))$ is contained in a closed ball of the normed space $H^1([a,b],\sR^N)$. Since the latter space is separable \cite[Theorem 3.6]{Adams}, that ball is weakly metrizable \cite[Theorem 6.30]{Aliprantis} and, on the strength of Alaoglu's theorem \cite[Theorem 6.21]{Aliprantis}, weakly compact. Hence, also $i_\ast(A^{[a,b]}_\T(\K_1,\K_2))$ is weakly metrizable, whereas proving its weak compactness amounts to showing that for any sequence $(\gamma_n) \subset A^{[a,b]}_\T(\K_1,\K_2)$ such that $i_\ast(\gamma_n) \rightharpoonup u$ in $H^1([a,b],\sR^N)$ one has $u = i_\ast(\gamma)$ for some $\gamma \in A^{[a,b]}_\T(\K_1,\K_2)$. Indeed, by Lemma \ref{lem_cont_rep} i) $i_\ast(\gamma_n) \rightarrow u$ in $C([a,b],\sR^N)$, and from Corollary \ref{C0viaNM} we obtain that $u = i_\ast(\gamma)$ for some $\gamma \in A^{[a,b]}_\T$. Of course, uniform convergence yields also that $\gamma(a) \in \K_1$ and $\gamma(b) \in \K_2$.

Finally, thanks to $i_\ast$ being an embedding (Remark \ref{rem_weakH1loctop}), we obtain the desired result.
\end{proof}

\begin{lem}
\label{lem_metrizable}
$A^I_\T$ is metrizable in the weak-$H^1_\loc$-topology.
\end{lem}
\begin{proof}
Assume first that $I \vc [a,b]$. Since $A^{[a,b]}_\T = \bigcup_{n,m} A^{[a,b]}_\T(\K_n,\K_m)$, where $\{\K_n\}$ is an exhaustion of $\M$ by compact sets, by Lemma \ref{lem_compact} we obtain that $A^{[a,b]}_\T$ is $\sigma$-compact.

Moreover, $A^{[a,b]}_\T$ is also locally compact and locally metrizable. Indeed, fix any $\gamma \in A^{[a,b]}_\T$ along with $t \in (a,b)$ and let $\U$ be an open precompact neighborhood of $\gamma(t)$ in $\M$. By Lemma \ref{lem_ev_contA} the set $\ev_t^{-1}(\U)$ is an open neighborhood of $\gamma$ in $A^{[a,b]}_\T$. Let us define the sets
\begin{align}
\label{KaKb}
\K_a \vc J^-(\overline{\U}) \cap \T^{-1}(a) \quad \textnormal \quad \K_b \vc J^+(\overline{\U}) \cap \T^{-1}(b),
\end{align}
which are compact by Proposition \ref{PropCCC}. Observe that $\ev_t^{-1}(\U) \subset A^{[a,b]}_\T(\K_a,\K_b)$. By Lemma \ref{lem_compact}, the latter is thus a compact and metrizable neighborhood of $\gamma$.

Of course $A^{[a,b]}_\T$ is a Hausdorff space (because $H^1_\loc([a,b],\sR^N)$ with the weak topology is). Now enter some general topological facts \cite{Willard}.
\begin{itemize}
\item Every locally compact Hausdorff space is regular.
\item Every $\sigma$-compact space is Lindel\"of.
\item Every regular Lindel\"of space is paracompact.
\item Every Hausdorff paracompact locally metrizable space is metrizable.
\end{itemize}
Hence $A^{[a,b]}_\T$ is metrizable in the weak-$H^1_\loc$-topology.

Let now $I$ be any interval and let $\{[a_m,b_m]\}$ be its exhaustion by compact subintervals. For every $m \in \sN$ let $d_m$ be a metric on $A^{[a_m,b_m]}_\T$ inducing the weak-$H^1_\loc$-topology, existing by the previous part of the proof. It is easy to convince oneself that the map $d: A^I_\T \times A^I_\T \rightarrow \sR$ defined via
\begin{align*}
d(\gamma,\rho) \vc \sum_{m \in \sN} \frac{1}{2^m} \frac{d_m(\gamma|_{[a_m,b_m]},\rho|_{[a_m,b_m]})}{1 + d_m(\gamma|_{[a_m,b_m]},\rho|_{[a_m,b_m]})}
\end{align*}
is a metric inducing the weak-$H^1_\loc$-topology on $A^I_\T$ (cf. Remark \ref{rem_weakH1loctop}), what concludes the proof.
\end{proof}

The fact that $A^I_\T$ is metrizable in the weak-$H^1_\loc$-topology may seem surprising --- after all, the spaces $H^1_\loc(I, \sR^N)$ and $H^1(I, \sR^N)$ are \emph{not} weakly metrizable or even weakly first countable (cf. e.g. \cite[Theorem 6.26]{Aliprantis}). This is extremely convenient, because it allows us to use sequences instead of cumbersome nets, which in turn enables us to show the announced equality of the compact-open topology and the weak-$H^1_\loc$-topology without much difficulty.

\begin{thm}
\label{stronger_top}
The weak-$H^1_\loc$-topology on $A^I_\T$ is equal to the compact-open topology.
\end{thm}
\begin{proof}
By the metrizability of both topologies, we want to show that $\gamma_n \xrightharpoonup{i} \gamma$ in $A^I_\T$ iff $\gamma_n \rightarrow \gamma$ $h$-uniformly on compact subsets. In light of Remark \ref{rem_weakH1loctop}, it suffices to consider the compact case $I = [a,b]$. Moreover, by the very definition of the weak-$H^1_\loc$-topology and the isometricity of the Nash--M\"uller embedding $i$, we can rewrite the desired equivalence as
\begin{align}
\label{stronger_top1}
i_\ast(\gamma_n) \rightharpoonup i_\ast(\gamma) \textnormal{ in } H^1_\loc([a,b],\sR^N) \quad \Longleftrightarrow \quad i_\ast(\gamma_n) \rightarrow i_\ast(\gamma) \textnormal{ uniformly on } [a,b]
\end{align}

In order to show `$\Rightarrow$', we claim first that the set $i_\ast( \bigcup_{n \in \sN} \{ \gamma_n \} \cup \{ \gamma \} )$ is $\|.\|_{H^1[a,b]}$-bounded. Indeed, fix $t \in (a,b)$ and recall that $\gamma_n(t) \rightarrow \gamma(t)$ in $\M$ by Lemma \ref{lem_ev_contA}. This means that $\bigcup_{n \in \sN} \{\gamma_n(t)\} \cup \{\gamma(t)\}$ is a compact set, and let $\U \subset \M$ be its open precompact superset. Defining now $\K_a, \K_b$ as in (\ref{KaKb}), we obtain that $\bigcup_{n \in \sN} \{\gamma_n\} \cup \{\gamma\} \subset A^{[a,b]}_\T(\K_a,\K_b)$, and so $i_\ast(\bigcup_{n \in \sN} \{\gamma_n\} \cup \{\gamma\}) \subset i_\ast(A^{[a,b]}_\T(\K_a,\K_b))$, which we know to be $\|.\|_{H^1[a,b]}$-bounded (cf. the beginning of the proof of Lemma \ref{lem_compact}).

Invoking now Lemma \ref{lem_bounded}, we deduce that $i_\ast(\gamma_n) \rightharpoonup i_\ast(\gamma)$ in $H^1([a,b],\sR^N)$, which by Lemma \ref{lem_cont_rep} i) means that $i_\ast(\gamma_n) \rightarrow i_\ast(\gamma)$ uniformly on $[a,b]$.

To show `$\Leftarrow$', notice that the set $i_\ast( \bigcup_{n \in \sN} \{ \gamma_n \} \cup \{ \gamma \} )$ is again $\|.\|_{H^1[a,b]}$-bounded by almost the same reasoning --- the only difference is that instead of Lemma \ref{lem_ev_contA} we simply infer the pointwise convergence from the uniform convergence. Therefore, every subsequence of $(i_\ast(\gamma_n))$ possesses a weakly convergent subsubsequence, whose limit must be $i_\ast(\gamma)$, i.e. the assumed uniform limit (cf. the distributional convergence argument in the proof of Lemma \ref{lem_cont_rep}).
\end{proof}

\begin{cor}
\label{prop_independent_NM}
The weak-$H^1_\loc$-topology on $A^I_\T$, being equal to the compact-open topology, does not depend on the choice of the Nash--M\"uller embedding.
\end{cor}

Since the compact-open topology and the weak-$H^1_\loc$-topology on $A_\T^I$ are the same, from now on we shall simply write ``$\gamma_n \rightarrow \gamma$ in $A_\T^I$'' instead of referring explicitly to any of them.

To finish this section, let us strengthen Lemma \ref{lem_ev_contA}.
\begin{prop}
\label{prop_ev_contB}
For any $t \in I$ the evaluation map $\ev_t: A^I_\T \rightarrow \M$ is continuous and proper.
\end{prop}
\begin{proof}
Continuity of $\ev_t$ is obvious. To see properness, take any compact $\K \subset \M$ and consider the closed set $\ev_t^{-1}(\K)$. In order to prove its compactness, let us take any sequence $(\gamma_n) \subset \ev_t^{-1}(\K)$ and find its subsequence convergent in $A^I_\T$.

To this end, let $\{[a_m,b_m]\}$ be an exhaustion of $I$ by compact subintervals. It is now crucial to observe that, for any fixed $m \in \sN$,
\begin{align*}
(\gamma_n|_{[a_m,b_m]})_n \subset A^{[a_m,b_m]}_\T\left( J^-(\K) \cap \T^{-1}(a_m), J^+(\K) \cap \T^{-1}(b_m) \right),
\end{align*}
where the latter set is a compact subset of $A^{[a_m,b_m]}_\T$ by Lemma \ref{lem_compact} \& Proposition \ref{PropCCC}. Bearing that in mind, we can find a convergent subsequence of $(\gamma_n)$ using the following version of the standard diagonal argument.

Firstly, let $(\gamma^{(1)}_n)_n$ be a subsequence of $(\gamma_n)$ such that $(\gamma^{(1)}_n|_{[a_1,b_1]})_n$ converges in $A^{[a_1,b_1]}_\T$. Then, inductively for $m = 2,3,\ldots$, let $(\gamma^{(m)}_n)_n$ be a subsequence of $(\gamma^{(m-1)}_n)_n$ such that $(\gamma^{(m)}_n|_{[a_m,b_m]})_n$ converges in $A^{[a_m,b_m]}_\T$. One now simply notices that the sequence $(\gamma^{(n)}_n)_n$, i.e. the ``diagonal'' subsequence of the sequence $(\gamma_n)$, has the property that $(\gamma^{(n)}_n|_{[a,b]})_n$ converges in $A^{[a,b]}_\T$ for any $[a,b] \subset I$. But this is equivalent to saying that $(\gamma^{(n)}_n)_n$ converges in $A^I_\T$, as desired.
\end{proof}

\begin{cor}
\label{prop_LCP}
$A^I_\T$ is a locally compact Polish space.
\end{cor}
\begin{proof}
To see the local compactness, let $\gamma \in A^I_\T$, $t \in I$ and let $\K \subset \M$ be a compact neighborhood of $\gamma(t)$. By Proposition \ref{prop_ev_contB}, $\ev_t^{-1}(\K)$ is a compact neighborhood of $\gamma$ in $A^I_\T$.

As for Polishness, it was already established by Proposition \ref{AIT_Polish}. However, one can also demonstrate it topologically. Indeed, notice first that by Proposition \ref{prop_ev_contB} $A^I_\T$ is $\sigma$-compact, because if $\{\K_n\}$ is an exhaustion of $\M$ by compact sets, then for any fixed $t \in I$ one has $A_\T^I = \bigcup_n \ev_t^{-1}(\K_n)$. Invoking Lemma \ref{lem_metrizable} we thus obtain that $A_\T^I$ is a metrizable, locally compact and $\sigma$-compact space, what already implies that it is Polish \cite[Theorem 3.4.1]{garling2017polish}.
\end{proof}


\section{On $L^p_\loc$-vector fields}
\label{sec::vector_fields}

In this section we investigate the notion of a $L^p_\loc$-vector field more closely. After providing the standard definition, we present an alternative, equivalent one, which does not involve local charts and will be later used in the proof of Theorem \ref{premain}. Let us start, however, by recalling the notion of $L^p_\loc$-spaces.

Let $\N$ be a topological space and let $\nu$ be a Borel measure on $\N$. For any $p \in [1,\infty]$, a Borel function $f: \N \rightarrow \sR$ is said to be locally $p$-$\nu$-integrable (or, in case $p=\infty$, locally $\nu$-essentially bounded) if for any open $\V \Subset \N$ $f|_\V \in L^p(\V,\nu)$. The spaces of such functions\footnote{Formally, the elements of $L^p_\loc(\N,\nu)$ are equivalence classes of functions equal $\nu$-a.e., but it is more convenient to regard them as maps defined up to $\nu$-null sets.} are denoted $L^p_\loc(\N,\nu)$ Notice that, by H\"older's inequality, for any $p_1, p_2 \in [1,\infty]$ with $p_1 \leq p_2$ one has $L^{p_1}_\loc(\N,\nu) \supset L^{p_2}_\loc(\N,\nu)$.

From now on, let $\N$ be a smooth manifold of dimension $D$.
\begin{defn}
\label{def_VF}
An $L^p_\loc(\nu)$-vector field on $\N$ is a map $X: C^\infty(\N) \rightarrow L^p_\loc(\N,\nu)$ such that, for any $\Psi \in C^\infty(\N)$, any chart $(\U, x)$ and any open $\V \Subset \U$
\begin{align}
\label{def_VF1}
(X \Psi)|_\V = \sum_{i=1}^D X^i \tfrac{\partial}{\partial x^i}\Psi|_\V
\end{align}
for some $X^i \in L^p(\V,\nu)$, $i = 1,\ldots,D$, where $\tfrac{\partial}{\partial x^i}: C^\infty(\V) \rightarrow C^\infty(\V)$ are the basis smooth vector fields associated with the chart $(\U, x)$.
\end{defn}

Notice that condition (\ref{def_VF1}) asserts that $X$ is linear and that $X \Psi$ indeed belongs to $L^p_\loc(\N,\nu)$. To see the latter, let $\{\V_j\}$ be a cover of $\N$ by open subsets, with each $\V_j$ compactly embedded in some chart domain, and let $\{\varphi_j\}$ be a partition of unity subordinate to that cover. For any open $\W \Subset \N$ one can pick a finite set of indices $J_\W$ such that $\sum_{j \in J_\W} \varphi_j|_\W = 1$ and write
\begin{align*}
& \|X \Psi\|_{L^p(\W)} = \left\|\sum_{j \in J_\W} \varphi_j X \Psi \right\|_{L^p(\W)} \leq \sum_{j \in J_\W} \|\varphi_j X \Psi\|_{L^p(\V_j)} 
\\
& \leq \sum_{j \in J_\W} \sum_{i=1}^D \|\varphi_j X^i \tfrac{\partial}{\partial x^i} \Psi\|_{L^p(\V_j)} \leq \sum_{j \in J_\W} \sum_{i=1}^D \left\| \varphi_j \tfrac{\partial}{\partial x^i} \Psi \right\|_{C(\overline{\V}_j)} \|X^i\|_{L^p(\V_j)} < +\infty,
\end{align*}
where $\left\| f \right\|_{C(\K)} \vc \sup_{y \in \K} |f(y)|$. Thus $(X \Psi)|_\W \in L^p(\W,\nu)$, and hence $X\Psi \in L^p_\loc(\N,\nu)$ by the arbitrariness of $\W$.

Condition (\ref{def_VF1}) has also another important consequence --- locality.
\begin{prop}
\label{prop_VF}
Let $X$ be a $L^p_\loc(\nu)$-vector field on $\N$. For any $\Psi \in C^\infty(\N)$, if $\Psi|_\W = 0$ for some open $\W \subset \N$, then $(X\Psi)|_\W = 0$, what can be expressed more succinctly as the inclusion $\supp X\Psi \subset \supp \Psi$.
\end{prop}
\begin{proof}
Take any $q \in \W$ and let $\V \ni q$ be an open precompact subset of $\W \cap \U$, where $(\U,x)$ is a chart. Then by (\ref{def_VF1})
$(X\Psi)|_\V = \sum_i X^i \tfrac{\partial}{\partial x^i}\Psi|_\V = 0$.
\end{proof}

\begin{cor}
\label{cor_VF}
For any $\Psi_1, \Psi_2 \in C^\infty(\N)$, if $\Psi_1|_\W = \Psi_2|_\W$ for some open $\W \subset \N$, then $(X\Psi_1)|_\W = (X\Psi_2)|_\W$.
\end{cor}
\begin{proof}
Rewrite the assumption as $(\Psi_1-\Psi_2)|_\W = 0$ and invoke Proposition \ref{prop_VF} along with the linearity of $X$.
\end{proof}

Let us now give an alternative, equivalent definition of a $L^p_\loc(\nu)$-vector field in the form of a theorem.
\begin{thm}
\label{thm_VF}
Let $\wv{X}: C^\infty_\cc(\N) \rightarrow L^p(\N,\nu)$ be a map satisfying the \emph{chain rule}
\begin{align}
\label{thm_VF1}
\wv{X}(\Theta(\Phi_1, \ldots, \Phi_L)) = \sum_{l=1}^L \partial_l \Theta(\Phi_1, \ldots, \Phi_L) \wv{X}\Phi_l
\end{align}
for every $\Theta \in C^\infty(\sR^L)$, $\Theta(0,\ldots,0)=0$ and every $\Phi_1, \ldots, \Phi_L \in C^\infty_{\cc}(\N)$. Then it can be uniquely extended to an $L^p_\loc(\nu)$-vector field $X$. 

Conversely, given an $L^p_\loc(\nu)$-vector field $X$ on $\N$, we have $X(C^\infty_\cc(\N)) \subset L^p(\N,\nu)$ and the restriction $\wv{X} \vc X|_{C^\infty_\cc(\N)}$ satisfies (\ref{thm_VF1}).
\end{thm}
\begin{proof}
First, notice that (\ref{thm_VF1}) in particular implies that $\wv{X}$ is linear and satisfies the Leibniz rule
\begin{align}
\label{thm_VF2}
\forall \Phi_1, \Phi_2 \in C^\infty_\cc(\N) \quad \wv{X}(\Phi_1 \Phi_2) = \wv{X}(\Phi_1) \Phi_2 + \Phi_1 \wv{X}(\Phi_2).
\end{align}
We claim that this already implies the locality of $\wv{X}$, that is
\begin{align}
\label{thm_VF3}
 \forall \Phi_1, \Phi_2 \in C^\infty_\cc(\N) \ \ \forall \, \textnormal{open } \W \subset \N \quad \ \Phi_1|_\W = \Phi_2|_\W \ \Rightarrow \ (\wv{X}\Phi_1)|_\W = (\wv{X}\Phi_2)|_\W.
\end{align}
Indeed, similarly as in Proposition \ref{prop_VF} and Corollary \ref{cor_VF}, it suffices to show that $\Phi|_\W = 0 \Rightarrow (\wv{X}\Phi)|_\W = 0$. To this end, pick any $q \in \W$ and let $\V \ni q$ be an open precompact subset of $\W$. Take any $\Omega \in C^\infty_\cc(\N)$ such that $0 \leq \Omega \leq 1$, $\Omega|_\V = 1$, $\Omega|_{\W^c} = 0$. The product $\Phi \Omega$ is the zero map and so, by the Leibniz rule,
\begin{align*}
\wv{X}(\Phi) \Omega + \Phi \wv{X}(\Omega) = 0, \quad \textnormal{and so} \quad (\wv{X}\Phi)|_\V = 0.
\end{align*}
By the arbitrariness of $q \in \W$, we thus obtain that $(\wv{X}\Phi)|_\W = 0$.

In order to extend the map $\wv{X}$ onto $C^\infty(\N)$, fix a cover $\{\V_j\}$ of $\N$ by open precompact sets along with a smooth partition of unity $\{\varphi_j\}$ subordinate to that cover (which in particular means that $\{\varphi_j\} \subset C^\infty_\cc(\N)$), and for any $\Psi \in C^\infty(\N)$ define
\begin{align}
\label{thm_VF4}
X\Psi \vc \sum_j \wv{X}(\varphi_j \Psi).
\end{align}
There are no problems with convergence, because the family of supports $\{ \supp \wv{X}(\varphi_j \Psi) \}_j$ is locally finite. Indeed, any point $q \in \N$ has an open precompact neighborhood $\V$ such that $\varphi_j|_\V = 0$ for all but finitely many $j$, hence also $(\varphi_j\Psi)|_\V = 0$ and, by (\ref{thm_VF3}), $\wv{X}(\varphi_j\Psi)|_\V = 0$ for all but finitely many $j$.

Notice, moreover, that the above definition of $X\Psi$ does not depend on the chosen partition of unity (or the cover). Indeed, if $\{\psi_k\} \subset C^\infty_\cc(\N)$ is another smooth partition of unity, then for any open $\W \Subset \N$
\begin{align*}
\big(\sum_{j \in J_\W} \varphi_j \Psi \big)\big|_\W = \Psi|_\W = \big(\sum_{k \in K_\W} \psi_k \Psi\big)\big|_\W,
\end{align*}
where $J_\W, K_\W$ are finite sets of indices such that $\sum_{j \in J_\W} \varphi_j|_\W = \sum_{k \in K_\W} \psi_k|_\W$ $= 1$. On the strength of (\ref{thm_VF3}), we thus obtain
\begin{align*}
\sum_{j \in J_\W} \wv{X}(\varphi_j \Psi)\big|_\W = \big(\wv{X}\big( \sum_{j \in J_\W} \varphi_j \Psi \big)\big)\big|_\W = \big(\wv{X}\big( \sum_{k \in K_\W} \psi_k \Psi \big)\big)\big|_\W = \sum_{k \in K_\W} \wv{X}(\psi_k \Psi)\big|_\W.
\end{align*}
Thus defined $X\Psi$ belongs to $L^p_\loc(\N,\nu)$, because for any open $\W \Subset \N$
\begin{align*}
\|X\Psi\|_{L^p(\W)} = \big\| \sum_j \wv{X}(\varphi_j \Psi) \big\|_{L^p(\W)} = \big\| \sum_{j \in J_\W} \wv{X}(\varphi_j \Psi) \big\|_{L^p(\W)} \leq \sum_{j \in J_\W} \| \wv{X}(\varphi_j \Psi) \|_{L^p(\W)} < +\infty.
\end{align*}

All in all, we have a well-defined map $X: C^\infty(\N) \rightarrow L^p_\loc(\N,\nu)$ extending $\wv{X}$. We claim that the Leibniz rule still holds for the extension. To demonstrate that, observe first that
\begin{align}
\label{thm_VF5}
X(1) = \sum_j \wv{X}(\varphi_j) = 0.
\end{align}
Indeed, since for any open $\W \Subset \N$ we have $\sum_j \varphi_j|_\W = \sum_{j \in J_\W} \varphi_j|_\W = 1$, therefore one can write that
\begin{align*}
\big( \sum_{j \in J_\W} \varphi_j \big)^2 \big|_\W = \sum_{j \in J_\W} \varphi_j|_\W,
\end{align*}
which, by (\ref{thm_VF3}), implies that
\begin{align*}
\wv{X}\big( \sum_{j \in J_\W} \varphi_j \big)^2 \big|_\W = \wv{X}\big(\sum_{j \in J_\W} \varphi_j\big)\big|_\W.
\end{align*}
Using now (\ref{thm_VF2}) and the fact that $\sum_{j \in J_\W} \varphi_j|_\W = 1$, one gets
\begin{align*}
2 \wv{X}\big( \sum_{j \in J_\W} \varphi_j \big) \big|_\W = \wv{X}\big(\sum_{j \in J_\W} \varphi_j\big)\big|_\W,
\end{align*}
and so 
\begin{align*}
\sum_{j} \wv{X}(\varphi_j)|_\W = \sum_{j \in J_\W} \wv{X}(\varphi_j)|_\W = \wv{X}\big( \sum_{j \in J_\W} \varphi_j \big) \big|_\W = 0.
\end{align*}
With thus proven (\ref{thm_VF5}) in mind, for any $\Psi_1, \Psi_2 \in C^\infty(\N)$ one obtains that
\begin{align*}
X(\Psi_1 \Psi_2) & = \sum_j \wv{X}(\varphi_j \Psi_1 \Psi_2) = \sum_{j,k} \varphi_k \wv{X}(\varphi_j \Psi_1 \Psi_2)
\\
& = \sum_{j,k} \wv{X}(\varphi_j \Psi_1 \varphi_k \Psi_2) - \sum_{j,k} \wv{X}(\varphi_k) \varphi_j \Psi_1 \Psi_2
\\
& = \sum_{j,k} \wv{X}(\varphi_j \Psi_1) \varphi_k \Psi_2 + \sum_{j,k} \varphi_j \Psi_1 \wv{X}(\varphi_k \Psi_2) - \Psi_1 \Psi_2 \sum_k \wv{X}(\varphi_k)
\\
& =  \Psi_2 \sum_j \wv{X}(\varphi_j \Psi_1) + \Psi_1 \sum_j \wv{X}(\varphi_k \Psi_2) = X(\Psi_1) \Psi_2 + \Psi_1 X(\Psi_2).
\end{align*}

Since the extension $X$ obeys the Leibniz rule, it is local in the sense of Corollary \ref{cor_VF} --- the proof goes along the same lines as the one of (\ref{thm_VF3}).

We are finally ready to prove that $X$ satisfies (\ref{def_VF1}). Fix $\Psi \in C^\infty(\N)$, a chart $(\U,x)$ and an open $\V \Subset \U$, where we assume (without loss of generality) that $(0,\ldots,0) \in x(\V)$. Let $\Omega \in C^\infty_\cc(\U)$ be such that $0 \leq \Omega \leq 1$, $\Omega|_\V = 1$ and consider the map $\widetilde{\Psi} \in C^\infty_\cc(\N)$ defined via
\begin{align*}
\widetilde{\Psi}(q) \vc \left\{ \begin{array}{ll} 
								\Psi(x^{-1}(\Omega(q) x(q))) - \Psi(x^{-1}(0,\ldots,0)) & \textnormal{for } q \in \U
								\\
								0 & \textnormal{otherwise.}
								\end{array} \right.
\end{align*}
Notice that $\Psi|_\V = (\widetilde{\Psi} + \Psi(x^{-1}(0,\ldots,0))|_\V$, and so, by the locality of $X$ and (\ref{thm_VF5}), we can write that
\begin{align*}
X(\Psi)|_\V = X(\widetilde{\Psi})|_\V = \wv{X}(\widetilde{\Psi})|_\V = \wv{X}(\Theta(\Omega x^1, \ldots, \Omega x^D))|_\V,
\end{align*}
where $\Theta \in C^\infty(\sR^D)$ is defined as 
\begin{align*}
\Theta(a^1,\ldots,a^D) \vc \Psi(x^{-1}(a^1,\ldots,a^D)) - \Psi(x^{-1}(0,\ldots,0)).
\end{align*}
Using now the chain rule (\ref{thm_VF1}), we obtain that
\begin{align*}
X(\Psi)|_\V = \big( \sum_{i=1}^D \partial_i \Theta(\Omega x^1, \ldots, \Omega x^D) \wv{X}(\Omega x^i) \big)\big|_\V,
\end{align*}
and so by defining $X^i \vc \wv{X}(\Omega x^i)|_\V \in L^p(\V,\nu)$, $i=1,\ldots,D$ and noticing that
\begin{align*}
(\partial_i \Theta(\Omega x^1, \ldots, \Omega x^D))|_\V = (\partial_i(\Psi \circ x^{-1})(x^1, \ldots, x^D))|_\V = \tfrac{\partial}{\partial x^i} \Psi|_\V,
\end{align*}
we indeed obtain (\ref{def_VF1}).

Conversely, assume that an $L^p_\loc(\nu)$-vector field $X$ is given. The inclusion $X(C^\infty_\cc(\N)) \subset L^p(\N,\nu)$ is a direct consequence of Proposition \ref{prop_VF}. Thus we only need to check that $\wv{X} \vc X|_{C^\infty_\cc(\N)}$ satisfies the chain rule (\ref{thm_VF1}).

It suffices to check it locally and actually one does not even have to restrict oneself to $C^\infty_\cc(\N)$. Indeed, fix any chart $(\U,x)$ and any open $\V \Subset \U$. For any $\Theta \in C^\infty(\sR^L)$ and any $\Psi_1, \ldots, \Psi_L \in C^\infty(\N)$ one has
\begin{align*}
& X(\Theta(\Psi_1, \ldots, \Psi_L))|_\V = \sum_{i=1}^D X^i \tfrac{\partial}{\partial x^i}\Theta(\Psi_1, \ldots, \Psi_L)|_\V
\\
& = \sum_{i=1}^D X^i \sum_{l=1}^L \partial_l \Theta(\Psi_1, \ldots, \Psi_L)|_\V \tfrac{\partial}{\partial x^i} \Psi_l|_\V
\\
& = \sum_{l=1}^L \partial_l \Theta(\Psi_1, \ldots, \Psi_L)|_\V \sum_{i=1}^D X^i \tfrac{\partial}{\partial x^i} \Psi_l|_\V = \big( \sum_{l=1}^L \partial_l \Theta(\Psi_1, \ldots, \Psi_L) X\Psi_l \big)\big|_\V,
\end{align*}
where we have used the fact that the basis vector fields do obey the chain rule.
\end{proof}

We have thus established that an $L^p_\loc(\nu)$-vector field can be regarded either as a map $X: C^\infty(\N) \rightarrow L^p_\loc(\N,\nu)$ satisfying (\ref{def_VF1}) or as a map $\wv{X}: C^\infty_\cc(\N) \rightarrow L^p(\N,\nu)$ satisfying (\ref{thm_VF1}). For the sake of proving Theorem \ref{premain} it is important to understand how the properties of $X$ listed in (iii) translate to the language of $\wv{X}$. The answer is given by the following two propositions.

\begin{prop}
\label{property1_VF}
Let $X$, $\wv{X}$ be the two equivalent descriptions of an $L^p_\loc(\nu)$-vector field on $\N$. Then for any $F \in C^\infty(\N)$ and $G \in L^p_\loc(\N,\nu)$
\begin{align*}
X(F) = G \quad \nu\textnormal{-a.e.} \quad \Longleftrightarrow \quad \forall \Phi \in C^\infty_\cc(\N) \quad \wv{X}(\Phi F) - \wv{X}(\Phi)F = \Phi G \quad \nu\textrm{-a.e.}
\end{align*}
\end{prop}
\begin{proof}
The implication ,,$\Rightarrow$'' follows trivially from the Leibniz rule. As for the converse implication, one can write, by formula (\ref{thm_VF4}) defining the extension, that
\begin{align}
\label{property1_VF1}
X(F) = \sum_j \wv{X}(\varphi_j F) = \sum_j (\wv{X}(\varphi_j)F + \varphi_j G) = F \sum_j \wv{X}(\varphi_j) + G \sum_j \varphi_j = G,
\end{align}
where in the last equality we have used (\ref{thm_VF5}).
\end{proof}

In particular, the above proposition says that the condition $X\T = 1$ can be equivalently expressed via $\wv{X}(\Phi \T) - \wv{X}(\Phi)\T = \Phi$ for all $\Phi \in C^\infty_\cc(\N)$.

\begin{prop}
\label{property2_VF}
Let $\N$ be a stably causal spacetime, and let $X$, $\wv{X}$ be the two equivalent descriptions of an $L^p_\loc(\nu)$-vector field on $\N$. The following conditions are equivalent:
\begin{enumerate}[i)]
\item $X$ is future-directed causal $\nu$-a.e.
\item For any temporal function $\T \quad X\T > 0 \quad \nu$-a.e.
\item For any smooth bounded causal function $f \quad Xf \geq 0 \quad \nu$-a.e.
\item For any smooth bounded causal function $f$ and any $\Phi \in C^\infty_\cc(\N)$, $\Phi \geq 0$ \\ $\wv{X}(\Phi f) - \wv{X}(\Phi)f \geq 0 \quad \nu$-a.e.
\end{enumerate}
\end{prop}
\begin{proof}
The equivalences $(i) \Leftrightarrow (ii) \Leftrightarrow (iii)$ are guaranteed by Lemma \ref{lem_CT}. As for the equivalence $(iii) \Leftrightarrow (iv)$, it becomes evident as soon as one writes, on the strength of Proposition \ref{property1_VF}, the identity $\wv{X}(\Phi f) - \wv{X}(\Phi)f = \Phi Xf$, valid for any $\Phi \in C^\infty_\cc(\N)$.
\end{proof}


\section{Proof of the main result}
\label{sec::proof}

Having acquainted ourselves with the $H^1_\loc$-side of the compact-open topology on $A^I_\T$ and with the $L^p_\loc$-vector fields, we are finally ready to prove Theorem \ref{premain}. Let us restate it here, recalling the notation and employing the alternative, chart-independent definition of an $L^{\infty}_\loc$-vector field.
\begin{thm}
\label{main}
Let $\M$ be a globally hyperbolic spacetime and let $\T: \M \rightarrow \sR$ be a Cauchy temporal function. Let $I$ be an interval and let us denote $\M_I \vc \T^{-1}(I)$ and $\M_t \vc \T^{-1}(t)$ for $t \in I$. For any map $\mu: I \rightarrow \Pf(\M_I)$, $t \mapsto \mu_t$ such that $\supp \, \mu_t \subset \M_t$ for every $t \in I$ (called a $\T$-evolution of a probability measure), the following conditions are equivalent:
\begin{enumerate}[(\itshape i)]
\item $\mu$ is \emph{causal} in the sense that
\begin{align}
\label{causal_measure_map}
\forall \, s,t \in I \quad s \leq t \ \Rightarrow \ \mu_s \preceq \mu_t.
\end{align}
\item There exists $\sigma \in \Pf(A_\T^I)$ such that $(\ev_t)_\sharp \sigma = \mu_t$ for every $t \in I$, where $\ev_t: A_\T^I \rightarrow \M$, $\ev_t(\gamma) \vc \gamma(t)$ denotes the evaluation map.
\item There exists a map $\wv{X}: C^\infty_{\cc}(\M_I^\circ) \rightarrow L^\infty(\M_I, \eta)$, where $\eta \vc \int_I\mu_t dt$, satisfying
\begin{itemize}
\item Chain rule: \ $\forall \Theta \in C^\infty(\sR^L), \, \Theta(0,\ldots,0) = 0 \ \ \forall \Phi_1, \ldots, \Phi_L \in C^\infty_{\cc}(\M_I^\circ)$ \begin{align*}
\wv{X}(\Theta(\Phi_1, \ldots, \Phi_L)) = \sum_{l=1}^L \partial_l \Theta(\Phi_1, \ldots, \Phi_L) \wv{X}\Phi_l.
\end{align*}
\item ``$X\T = 1 \ \eta\textit{-a.e.}$'': \quad \ $\forall \Phi \in C^\infty_\cc(\M_I^\circ) \quad \quad \wv{X}(\Phi \T) - \wv{X}(\Phi)\T = \Phi.$
\item Continuity equation: \ $\forall \Phi \in C^\infty_{\cc}(\M_I^\circ) \quad \quad \int_{\M_I} \wv{X}\Phi \, d\eta = 0$.
\item Causality $\eta$-a.e.: \quad \ $\forall \Phi \in C^\infty_{\cc}(\M_I^\circ), \, \Phi \geq 0 \quad \forall \, \textnormal{causal } f \in C^\infty_\cb(\M_I^\circ)$ 
\begin{align*}
\wv{X}(\Phi f) - \wv{X}(\Phi) f \geq 0.
\end{align*}
\end{itemize}
What is more, for any $\Phi \in C^\infty_\cc(\M_I^\circ)$ the maps $\Lambda_\mu(\Phi)(t) \vc \int_{\M_I} \Phi \, d\mu_t$ and $\Lambda_\mu(\wv{X}\Phi)(t) \vc \int_{\M_I} \wv{X}\Phi \, d\mu_t$ are well defined for $t \in I$ a.e. and we have that $\Lambda_\mu(\Phi) \in W^{1,\infty}(I)$ along with $\Lambda_\mu(\Phi)' = \Lambda_\mu(\wv{X}\Phi)$.
\end{enumerate}
\end{thm}

Before delving into the proof, we need to reintroduce certain tools from \cite{Miller17a} (cf. also \cite{NParticleMEHH}).

First we discuss curve concatenation. For any fixed $a,b,c \in \sR$, $a<b<c$ consider the set $A^{[a,b,c]}_\T \vc \{ (\gamma_1, \gamma_2) \in A^{[a,b]}_\T \times A^{[b,c]}_\T \, | \, \gamma_1(b) = \gamma_2(b) \}$ of concatenable causal curves, which is a closed (and hence Polish) subspace of $A^{[a,b]}_\T \times A^{[b,c]}_\T$. The concatenation map $\sqcup: A^{[a,b,c]}_\T \rightarrow A^{[a,c]}_\T$ is of course continuous, and thus it can be extended to measures on the spaces of causal curves in the following way (cf. \cite[Definition 25]{NParticleMEHH}).
\begin{defn}
\label{def_concat}
For any measures $\sigma_1 \in \Pf(A^{[a,b]}_\T)$ and $\sigma_2 \in \Pf(A^{[b,c]}_\T)$ we say they are \emph{concatenable} if $(\ev_b)_\sharp \sigma_1 = (\ev_b)_\sharp \sigma_2 \cv \nu$ and we define their \emph{concatenation} $\sigma_1 \sqcup \sigma_2 \in \Pf(A^{[a,c]}_\T)$ with the help of the Riesz--Markov representation theorem (cf. \cite[Theorem 14.12]{Aliprantis}) via
\begin{align*}
\int_{A^{[a,c]}_\T} F d(\sigma_1 \sqcup \sigma_2) \vc \int_{\M_b} \left( \int_{A^{[a,b,c]}_\T} F(\gamma_1 \sqcup \gamma_2) \, d(\sigma^q_1 \times \sigma^q_2) (\gamma_1,\gamma_2) \right) d\nu(q),
\end{align*}
for any $F \in C_\cc(A^{[a,c]}_\T)$, where $\{\sigma^q_i \}_{q \in \T^{-1}(b)}$ is the disintegration of $\sigma_i$ with respect to $\ev_b$, $i=1,2$.
\end{defn}

Clearly, the functional on $C_\cc(A^{[a,c]}_\T)$ defined by the above iterated integral is linear, positive and of norm one. In invoking the Riesz--Markov theorem we rely on the local compactness of the spaces $A^I_\T$ (Proposition \ref{prop_LCP}).

Of course, one can similarly define the concatenation of measures in the case where one or both spaces of curves involve noncompact intervals. One can also easily verify that
\begin{align}
\label{rem_concat}
(\ev_t)_\sharp (\sigma_1 \sqcup \sigma_2) = \left\{ \begin{array}{ll} (\ev_t)_\sharp \sigma_1 & \textrm{for } t < b
																	 \\ \nu & \textrm{for } t = b
																	 \\ (\ev_t)_\sharp \sigma_2 & \textrm{for } t > b
													  \end{array} \right.
\end{align}

We will also need the following lemma (cf. \cite[Lemma 24]{NParticleMEHH}).
\begin{lem}
\label{lem_properonto}
The map $(\ev_a, \ev_b): A^{[a,b]}_\T \rightarrow J^+ \cap (\M_a \times \M_b)$, defined via $(\ev_a, \ev_b)(\gamma) \vc (\gamma(a), \gamma(b))$ is a continuous proper surjection admitting a Borel right inverse.
\end{lem}
\begin{proof}
Continuity follows immediately from Proposition \ref{prop_ev_contB}. 

In order to show surjectiveness, take any $p,q \in \M$ such that $\T(p) = a$, $\T(q) = b$ and $p \preceq q$. Let $\rho: [0,1] \rightarrow \M$ be a piecewise-smooth future-directed causal curve connecting $p$ with $q$. Reasoning similarly as in the proof of Proposition \ref{Liploc}, one finds that the reparametrized curve $\hat{\rho} \vc \rho \circ (\T \circ \rho)^{-1}$ belongs to $A^{[a,b]}_\T$ and indeed satisfies $(\ev_a, \ev_b)(\hat{\rho}) = (p,q)$.

Finally, to show properness, take any compact $\K \subset J^+ \cap (\M_a \times \M_b)$ and notice that
\begin{align*}
(\ev_a, \ev_b)^{-1}(\K) \subset \ev_a^{-1}(\pi^1(\K)) \cap \ev_b^{-1}(\pi^2(\K)),
\end{align*}
where $\pi^1, \pi^2: \M^2 \rightarrow \M$ are the canonical projections on the first and second component, respectively. On the strength of Lemma \ref{prop_ev_contB}, this means that $(\ev_a, \ev_b)^{-1}(\K)$ is compact as a closed subset of a compact set.

To finish the proof, we invoke the standard measurable selection result, by which a continuous map from a $\sigma$-compact metrizable space onto a metrizable space admits a Borel right inverse \cite[Corollary I.8]{Fabec}. Metrizability and $\sigma$-compactness of $A^I_\T$ were proven in Lemma \ref{lem_metrizable} \& Corollary \ref{prop_LCP}.
\end{proof}

Finally, the following lemma will be useful more than once.
\begin{lem}
\label{useful}
For any function $\Psi: \M \rightarrow \sR$ and any $\gamma \in A_\T^I$
\begin{align}
\label{useful1}
\supp \Psi \circ \gamma \subset \T(\supp \Psi).
\end{align}
\end{lem}
\begin{proof}
Pick any $t \in \supp \Psi \circ \gamma$ and let $(t_n) \subset I$ be a sequence convergent to $t$ such that $\Psi(\gamma(t_n)) \neq 0$ for every $n \in \sN$. The latter means that $\gamma(t_n) \in \supp \Psi$, and passing to the limit yields $\gamma(t) \in \supp \Psi$. This, in turn, leads to $t = \T(\gamma(t)) \in \T(\supp \Psi)$.
\end{proof}

We now proceed to proving Theorem \ref{main}.
\begin{proof}[Proof of Theorem \ref{main}]
\emph{(ii) $\Rightarrow$ (i)}: Fix $s,t \in I$, $s < t$. Similarly as in Lemma \ref{lem_properonto}, consider the map $(\ev_s, \ev_t): A^I_\T \rightarrow \M^2$ and define the pushforward measure $\omega \vc (\ev_s, \ev_t)_\sharp \sigma$. We claim that $\omega$ is a causal coupling of $\mu_s$ and $\mu_t$. Indeed, one has that
\begin{align*}
(\pi^1)_\sharp \omega = [\pi^1 \circ (\ev_s, \ev_t)]_\sharp \sigma = (\ev_s)_\sharp \sigma = \mu_s
\end{align*}
and similarly $(\pi^2)_\sharp \omega = \mu_t$. Moreover,
\begin{align*}
\omega(J^+) = \sigma((\ev_s, \ev_t)^{-1}(J^+)) = \sigma(A^I_\T) = 1,
\end{align*}
where we have used the fact that the image of the map $(\ev_s, \ev_t)$ is a subset of $J^+$ (cf. the first part of the proof of Lemma \ref{lem_properonto}).
\\

\emph{(i) $\Rightarrow$ (ii)}: 
Assume first that $I = [a,b]$. We shall construct a sequence of measures $(\sigma_n) \subset \Pf(A^{[a,b]}_\T)$ such that $(\ev_t)_\sharp \sigma_n = \mu_t$ for all $t$ of the form $t^n_i \vc a + i(b-a)/2^n$, $i=0,1,2,3,\ldots,2^n$, which will then turn out to have a subsequence convergent to some $\sigma \in \Pf(A^{[a,b]}_\T)$. Since condition (\ref{causal_measure_map}) implies narrow continuity \cite[Proposition 11]{Miller17a}, such a $\sigma$ must in fact satisfy $(\ev_t)_\sharp \sigma = \mu_t$ for \emph{all} $t \in [a,b]$, as desired.

We construct the sequence $(\sigma_n)$ as follows. For any fixed $n$ and any $i=1,2,3,\ldots,2^n$, let $E^i: J^+ \cap (\M_{t^n_{i-1}} \times \M_{t^n_i}) \rightarrow A^{[t^n_{i-1},t^n_i]}_\T$ be the Borel inverse of the map $(\ev_{t^n_{i-1}}, \ev_{t^n_i})$, existing by Lemma \ref{lem_properonto}. Furthermore, let $\omega_i$ be a causal coupling of $\mu_{t^n_{i-1}}$ and $\mu_{t^n_i}$, that is $\omega_i \in \itPi_\preceq(\mu_{t^n_{i-1}}, \mu_{t^n_i})$. Notice that each $\omega_i$ can be regarded as an element of $\Pf(J^+ \cap (\M_{t^n_{i-1}} \times \M_{t^n_i}))$. Using the concatenation introduced in Definition \ref{def_concat}, we can thus define
\begin{align*}
\sigma_n \vc E^1_\sharp \omega_1 \sqcup E^2_\sharp \omega_2 \sqcup E^3_\sharp \omega_3 \sqcup \ldots \sqcup E^{2^n}_\sharp \omega_{2^n} \in \Pf(A^{[a,b]}_\T).
\end{align*}
One can easily verify that indeed $(\ev_{t^n_i})_\sharp \sigma_n = \mu_{t^n_i}$ for every $i = 0,1,\ldots,2^n$. Similarly as in the proof of (ii) $\Rightarrow$ (i), one can also check that $(\ev_a, \ev_b)_\sharp \sigma_n \in \itPi_\preceq(\mu_a,\mu_b)$, and so the constructed sequence $(\sigma_n) \subset (\ev_a, \ev_b)^{-1}_\sharp(\itPi_\preceq(\mu_a,\mu_b))$.

Now comes the crucial observation: by the continuity and properness of $(\ev_a, \ev_b)$ (Lemma \ref{lem_properonto}) and the subsequent properness of $(\ev_a, \ev_b)_\sharp$ (cf. \cite[Lemma 6]{Miller17a}) along with the narrow compactness of $\itPi_\preceq(\mu_a,\mu_b)$ (Proposition \ref{narrowcom}), the set $(\ev_a, \ev_b)^{-1}_\sharp(\itPi_\preceq(\mu_a,\mu_b))$ is compact, and thus $(\sigma_n)$ has a convergent subsequence. Its limit $\sigma$, as already explained, is the desired measure on the space $A^{[a,b]}_\T$.

Assume now that $I$ is an arbitrary interval and let $\{[a_n,b_n]\}$ be its exhaustion by compact subintervals. Fix $t_0 \in (a_1,b_1)$ and notice that the inverse image $(\ev_{t_0})^{-1}_\sharp(\mu_{t_0}) \subset \Pf(A^I_\T)$ is narrowly compact (again by Proposition \ref{narrowcom} along with the continuity and properness of $\ev_{t_0}$ guaranteed by Proposition \ref{prop_ev_contB}). Somewhat similarly as before, we shall construct a sequence $(\sigma_n) \subset (\ev_{t_0})^{-1}_\sharp(\mu_{t_0})$ such that for every $n$ we will have $(\ev_t)_\sharp \sigma_n = \mu_t$ for all $t \in [a_n,b_n]$. By the compactness of $(\ev_{t_0})^{-1}_\sharp(\mu_{t_0})$, the sequence $(\sigma_n)$ has a subsequence convergent to some $\sigma \in \Pf(A^I_\T)$, which clearly satisfies the desired equality for all $t \in I$. 
 
The construction goes as follows. For any fixed $n$, let $I^1_n \vc I \cap (-\infty, a_n]$, $I^2_n \vc [a_n, b_n]$, $I^3_n \vc I \cap [b_n,+\infty)$, and let $\sigma_n$ be the concatenation of three measures\footnote{Of course, if the intervals $I^1_n$ and/or $I^3_n$ are degenerate or empty, one can disregard them.}
\begin{align*}
\sigma_n \vc \sigma^1_n \sqcup \sigma^2_n \sqcup \sigma^3_n,
\end{align*}
where $\sigma^j_n \in A^{I^j_n}_\T$, $j=1,2,3$ are defined in the following way:
\begin{itemize}
\item $\sigma^1_n \vc \int_{\M_{a_n}} \delta_{\gamma^q} d\mu_{a_n}(q)$, where for any $q \in \M_{a_n}$ the curve $\gamma^q: I^1_n \rightarrow \M$ is defined with the help of the GBS-splitting $\Xi: \M \rightarrow \sR \times \itSigma$ as $\gamma^q(t) \vc \Xi^{-1}(t,\pi_\itSigma(\Xi(q)))$. In other words, $\gamma^q$ is a worldline of a pointlike particle being ``at rest at $\pi_\itSigma(\Xi(q))$''.
\item $\sigma^2_n$ is the measure constructed on $A^{[a_n,b_n]}_\T$ as decsribed in the first part of the proof. As such, it safisfies $(\ev_t)_\sharp \sigma^2_n = \mu_t$ for all $t \in [a_n,b_n]$ and this property transfers to $\sigma_n$ on the strength of (\ref{rem_concat}).
\item $\sigma^3_n \vc \int_{\M_{b_n}} \delta_{\gamma_p} d\mu_{b_n}(p)$, where for any $p \in \M_{b_n}$ the curve $\gamma_p: I^3_n \rightarrow \M$ is defined with the help of the GBS-splitting $\Xi$ as $\gamma_p(t) \vc \Xi^{-1}(t,\pi_\itSigma(\Xi(p)))$. Just like the curves $\gamma^q$ above, every $\gamma_p$ is a worldline of a pointlike particle being ``at rest at $\pi_\itSigma(\Xi(p))$''.
\end{itemize}
The sole aim of $\sigma^1_n$ and $\sigma^3_n$ is to ``causally extend'' the curves encapsulated in $\sigma^2_n$, so as to obtain a well-defined measure $\sigma_n$ on $A^I_\T$. As already explained, the sequence $(\sigma_n)$ has a convergent subsequence, and its limit $\sigma$ is the desired measure on $A^I_\T$ satisfying $(\ev_t)_\sharp \sigma = \mu_t$ for all $t \in I$.
\\

\emph{(ii) $\Rightarrow$ (iii)} The idea is to define the map $\wv{X}: C^\infty_{\cc}(\M_I^\circ) \rightarrow L^2(\M_I, \eta)$ from $\sigma \in \Pf(A^I_\T)$ through the Riesz representation theorem via
\begin{align}
\label{def_X0}
\langle \wv{X} \Phi, \varphi \rangle_{L^2(\M_I)} \vc \int_{A^I_\T} \int_I (\Phi \circ \gamma)'(t) (\varphi \circ \gamma)(t) dt \, d\sigma(\gamma)
\end{align}
for any $\Phi \in C^\infty_{\cc}(\M_I^\circ)$ and any $\varphi \in L^2(\M_I, \eta)$. The task is now to demonstrate that the above integral is well defined, and then to show that thus defined $\wv{X}$ has all the properties listed in Theorem \ref{main}. This, on the strength of Theorem \ref{thm_VF} and Remark \ref{Linfrem}, will in turn certify that $\wv{X}\Phi$ in fact belongs to $L^\infty(\M_I, \eta)$ for any $\Phi \in C^\infty_{\cc}(\M_I^\circ)$.

Notice that the (finite) $L^2$-norm of $\varphi$ can be equivalently written as
\begin{align}
\nonumber
& \| \varphi \|_{L^2(\M_I)}^2 \vc \int_{\M_I} |\varphi|^2 d\eta = \int_I \int_{\M_I} |\varphi(q)|^2 d\mu_t(q) \, dt = \int_I \int_{A^I_\T} |\varphi(\gamma(t))|^2 d\sigma(\gamma) \, dt 
\\
\label{def_X2}
& = \int_{A^I_\T} \int_I |\varphi(\gamma(t))|^2 dt \, d\sigma(\gamma) = \int_{A^I_\T} \| \varphi \circ \gamma \|^2_{L^2(I)} d\sigma(\gamma),
\end{align}
where in the new line we have used Fubini's theorem to change the order of integration. Notice also that (\ref{def_X0}) can be written as
\begin{align}
\label{def_X1}
\langle \wv{X} \Phi, \varphi \rangle_{L^2(\M_I)} \vc \int_{A^I_\T} \langle (\Phi \circ \gamma)', \varphi \circ \gamma \rangle_{L^2(I)} d\sigma(\gamma).
\end{align}
Therefore, we have to show that the map $\gamma \mapsto \langle (\Phi \circ \gamma)', \varphi \circ \gamma \rangle_{L^2(I)}$ is $\sigma$-mea\-sur\-able and $\sigma$-integrable.

The proof of measurability goes in three steps.

\textbf{Step 1.} Let $\varphi \in C^\infty_\cc(\M_I^\circ)$. We claim that the map $\varphi_\ast: A^I_\T \rightarrow L^2(I)$ given by $\varphi_\ast(\gamma) \vc \varphi \circ \gamma$ is well defined and continuous.

Indeed, observe first that for any $\gamma \in A^I_\T$ the map $\varphi \circ \gamma$ belongs to $C_\cc(I)$ and so it is square-integrable.

Assume now that $\gamma_n \rightarrow \gamma$ in $A^I_\T$. With the aid of Lemma \ref{useful}, take $[a,b] \subset I^\circ$ such that $(a,b) \supset \T(\supp \varphi) \supset \supp \varphi \circ \rho$ for any $\rho \in A^I_\T$. One can write that
\begin{align*}
& \|\varphi_\ast(\gamma_n) - \varphi_\ast(\gamma)\|^2_{L^2(I)} = \int_I |\varphi(\gamma_n(t)) - \varphi(\gamma(t))|^2 dt
\\
& = \int_a^b |\varphi(\gamma_n(t)) - \varphi(\gamma(t))|^2 dt \leq L_\varphi^2 (b-a) \max_{t \in [a,b]} d_h(\gamma_n(t),\gamma(t))^2 \rightarrow 0,
\end{align*}
where $L_\varphi$ is the Lipschitz constant of $\varphi$ with respect to the distance $d_h$.

\textbf{Step 2.} Let $\Phi \in C^\infty_\cc(\M_I^\circ)$. We claim that the map $A^I_\T \ni \gamma \mapsto (\Phi \circ \gamma)' \in L^2(I)$ is well defined with a norm-bounded image, and that $\gamma_n \rightarrow \gamma$ in $A^I_\T$ implies $(\Phi \circ \gamma_n)' \rightharpoonup (\Phi \circ \gamma)'$ in $L^2(I)$.

Firstly, notice that $\Phi \circ \gamma$ is Lipschitz. Indeed, it is locally Lipschitz by Proposition \ref{Liploc} as well as compactly supported, because by Lemma \ref{useful} $\supp \Phi \circ \gamma \subset \T(\supp \Phi)$. Hence $(\Phi \circ \gamma)' \in L^\infty(I) \subset L^2(I)$.

Secondly, in order to show that the set $\{(\Phi \circ \gamma)' \ | \ \gamma \in A^I_\T\}$ is $\|.\|_{L^2(I)}$-bounded, observe that, for any $\gamma \in A^I_\T$ and any $t \in I$ for which $\gamma'(t)$ exists, one can write $(\Phi \circ \gamma)'(t) = h(\nabla_h \Phi_{\gamma(t)},\gamma'(t))$, where $h$ is the Riemannian metric given by (\ref{R_metric}) and $\nabla_h \Phi$ denotes the $h$-gradient of $\Phi$. Choosing $[a,b] \subset I^\circ$ such that the compact set $\T(\supp \Phi) \subset [a,b]$, one can estimate
\begin{align}
\nonumber
\|(\Phi \circ \gamma)' \|_{L^2(I)}^2 & = \int_a^b |(\Phi \circ \gamma)'(t)|^2 dt = \int_I h(\nabla_h \Phi_{\gamma(t)}, \gamma'(t))^2 dt
\\
\label{X_step2}
& \leq \int_a^b h(\nabla_h \Phi_{\gamma(t)},\nabla_h \Phi_{\gamma(t)}) \, h(\gamma'(t),\gamma'(t)) \, dt
\\
\nonumber
& \leq \int_a^b h(\nabla_h \Phi_{\gamma(t)},\nabla_h \Phi_{\gamma(t)}) \, 2\theta(\gamma(t))\alpha(\gamma(t)) \, dt
\\
\nonumber
& \leq 2(b-a)\max_{q \in \supp \Phi} h(\nabla_h \Phi_q,\nabla_h \Phi_q) \, \theta(q)\alpha(q),
\end{align}
where we have used (\ref{derbound}) and strongly relied on the compactness of $\supp \Phi$. Notice that the obtained bound is independent of $\gamma$.

Thirdly, assume that $\gamma_n \rightarrow \gamma$ in $A_\T^I$, which means that $i_\ast(\gamma_n) \rightharpoonup i_\ast(\gamma)$ in $H^1_\loc(I,\sR^N)$ for any fixed Nash--M\"uller embedding $i: \M \hookrightarrow \sR^N$. In order to show that $(\Phi \circ \gamma_n)' \rightharpoonup (\Phi \circ \gamma)'$ in $L^2(I)$, consider the smooth map $\Phi \circ i^{-1}: i(\M) \rightarrow \sR$. Since $i(\M)$ is closed, the embedding $i$ is proper \cite[Proposition 5.5]{Lee}, and so there exists a smooth extension $\widehat{\Phi \circ i^{-1}}: \sR^N \rightarrow \sR$ \cite[Lemma 5.34]{Lee}. This allows us to write, invoking Lemma \ref{Psilemma}, that
\begin{align*}
\Phi \circ \gamma_n = (\widehat{\Phi \circ i^{-1}})_\ast(i_\ast(\gamma_n)) \rightharpoonup (\widehat{\Phi \circ i^{-1}})_\ast(i_\ast(\gamma)) = \Phi \circ \gamma
\end{align*}
in $H^1_\loc(I)$. This in turn yields $(\Phi \circ \gamma_n)' \rightharpoonup (\Phi \circ \gamma)'$ in $L^2_\loc(I)$ by Lemma \ref{lem_diff_cont}. Finally, on the strength of Lemma \ref{lem_bounded} and the $\|.\|_{L^2(I)}$-boundedness proven above, we obtain that $(\Phi \circ \gamma_n)' \rightharpoonup (\Phi \circ \gamma)'$ in $L^2(I)$, as desired.

\textbf{Step 3.} Fix $\Phi \in C^\infty_\cc(\M_I^\circ)$ and $\varphi \in L^2(\M_I, \eta)$ and let $(\varphi_n) \subset C^\infty_{\cc}(\M_I^\circ)$ be such that $\varphi_n \rightarrow \varphi$ in $L^2(\M_I, \eta)$. By the previous two steps, we obtain that for any $n \in \sN$ the map
\begin{align*}
A^I_\T \ni \gamma \mapsto  \langle (\Phi \circ \gamma)', \varphi_n \circ \gamma \rangle_{L^2(I)}
\end{align*}
is continuous and hence $\sigma$-measurable. Since the pointwise limit of a sequence of measurable functions is measurable, we now only need to show that for $\sigma$-a.a. $\gamma \in A^I_\T$ there exists a subsequence $(\varphi_{n_k})$ such that
\begin{align}
\label{X_step3}
\lim_{k \rightarrow \infty} \langle (\Phi \circ \gamma)', \varphi_{n_k} \circ \gamma - \varphi \circ \gamma \rangle_{L^2(I)} = 0.
\end{align}
Indeed, observe that $\int_{A^I_\T} \| (\varphi_n \circ \gamma) - (\varphi \circ \gamma) \|_{L^2(I)}^2 d\sigma(\gamma) = \int_{\M_I} |\varphi_n - \varphi|^2 d\eta \rightarrow 0$ by the very definition of $(\varphi_n)$, where the equality follows by the same reasoning as in (\ref{def_X2}). But this means that for some subsequence $(\varphi_{n_k})$ we have $\| (\varphi_{n_k} \circ \gamma) - (\varphi \circ \gamma) \|_{L^2(I)} \rightarrow 0$ for $\sigma$-a.a $\gamma \in A^I_T$, which in turn yields (\ref{X_step3}) by the Cauchy--Schwarz inequality.

We have thus proven that the map $\gamma \mapsto \langle (\Phi \circ \gamma)', \varphi \circ \gamma \rangle_{L^2(I)}$ is $\sigma$-mea\-sur\-able. To show its $\sigma$-integrability, observe that
\begin{align*}
\int_{A^I_\T} | \langle (\Phi \circ \gamma)', \varphi \circ \gamma \rangle_{L^2(I)} | d\sigma(\gamma) \leq \tfrac{1}{2} \int_{A^I_\T} \|(\Phi \circ \gamma)' \|_{L^2(I)}^2 d\sigma(\gamma) + \tfrac{1}{2} \int_{A^I_\T} \| \varphi \circ \gamma \|_{L^2(I)}^2 d\sigma(\gamma).
\end{align*}
The finiteness of the first term on the right-hand side has been already demonstrated through (\ref{X_step2}). As for the second term, notice that it is equal to $\tfrac{1}{2} \| \varphi \|^2_{L^2(\M_I)}$ (cf. (\ref{def_X2})) and is hence finite as well. 

All in all, we have thus shown that the integral in (\ref{def_X1}) exists for any $\Phi \in C^\infty_\cc(\M_I^\circ)$ and any $\varphi \in L^2(\M_I, \eta)$ and so the map $\wv{X}: C^\infty_{\cc}(\M_I^\circ) \rightarrow L^2(\M_I, \eta)$ is well defined.

We now demonstrate that $\wv{X}$ satisfies all the properties listed in Theorem \ref{main} (iii).

\smallskip

\textbf{Chain rule:} For any $\Theta \in C^\infty(\sR^L)$, any $\Phi_1,\ldots,\Phi_L \in C^\infty_{\cc}(\M_I^\circ)$ and any $\varphi \in L^2(\M_I, \eta)$ one has that
\begin{align*}
\langle \wv{X}(\Theta(\Phi_1,\ldots,\Phi_L)), \varphi \rangle_{L^2(\M_I)} & \vc \int_{A^I_\T} \big\langle [\Theta(\Phi_1 \circ \gamma,\ldots,\Phi_L \circ \gamma)]', \varphi \circ \gamma \big\rangle_{L^2(I)} d\sigma(\gamma)
\\
& = \sum_{l=1}^L \int_{A^I_\T} \big\langle \partial_l \Theta(\Phi_1 \circ \gamma,\ldots,\Phi_L \circ \gamma) (\Phi_l \circ \gamma)', \varphi \circ \gamma \big\rangle_{L^2(I)} d\sigma(\gamma)
\\
& = \sum_{l=1}^L \big\langle \partial_l \Theta(\Phi_1 \circ \gamma, \ldots, \Phi_L \circ \gamma) \wv{X}\Phi_l, \varphi \big\rangle_{L^2(\M_I)} \\
& = \big\langle \sum_{l=1}^L \partial_l \Theta(\Phi_1 \circ \gamma,\ldots,\Phi_L \circ \gamma) \wv{X}\Phi_l, \varphi \big\rangle_{L^2(\M_I)}.
\end{align*}

Before moving on, recall from the proof of Theorem \ref{thm_VF} that the chain rule entails the inclusion 
\begin{align}
\label{support_inclusion}
\forall \Phi \in C^\infty_{\cc}(\M_I^\circ) \quad \supp \wv{X}\Phi \subset \supp \Phi.
\end{align}

\textbf{``$X\T = 1$ $\eta$-a.e.'':}  Fix $\Phi \in C^\infty_{\cc}(\M_I^\circ)$ and $\varphi \in L^2(\M_I, \eta)$. Thanks to (\ref{support_inclusion}), the product $\wv{X}(\Phi)\T$ is a well-defined element of $L^2(\M_I, \eta)$. Remembering that $\T \circ \gamma = \id_I$ for any $\gamma \in A^I_\T$, one can write
\begin{align*}
\langle \wv{X}(\Phi \T) - \wv{X}(\Phi)\T, \varphi \rangle_{L^2(\M_I)} & = \int_{A^I_\T} \langle [(\Phi \circ \gamma)(\T \circ \gamma)]' - (\Phi \circ \gamma)'(\T \circ \gamma), \varphi \circ \gamma \rangle_{L^2(I)} d\sigma(\gamma)
\\
& = \int_{A^I_\T} \langle (\Phi \circ \gamma)(\T \circ \gamma)', \varphi \circ \gamma \rangle_{L^2(I)} d\sigma(\gamma)
\\
& = \int_{A^I_\T} \langle \Phi \circ \gamma, \varphi \circ \gamma \rangle_{L^2(I)} d\sigma(\gamma) = \langle \Phi, \varphi \rangle_{L^2(\M_I)}.
\end{align*}

\textbf{Continuity equation:} Fix $\Phi \in C^\infty_{\cc}(\M_I^\circ)$ and, with the help of Lemma \ref{useful}, take $[a,b] \subset I^\circ$ such that $(a,b) \supset \T(\supp \Phi) \supset \supp \Phi \circ \gamma$ for any $\gamma \in A^I_\T$. By (\ref{support_inclusion}), one can write that
\begin{align*}
\int_{\M_I} \wv{X} \Phi \, d\eta & = \int_{\M_I} \wv{X}\Phi \, \chi_{\supp \Phi} \, d\eta = \langle \wv{X} \Phi, \chi_{\supp \Phi} \rangle_{L^2(\M_I)}
\\
& = \int_{A^I_\T} \int_I (\Phi \circ \gamma)'(t) \, \chi_{\supp \Phi}(\gamma(t)) \, dt d\sigma(\gamma)
\\
& = \int_{A^I_\T} \int_I (\Phi \circ \gamma)'(t) \, \chi_{\supp \Phi \circ \gamma}(t) \, dt d\sigma(\gamma)
\\
& = \int_{A^I_\T} \int_a^b (\Phi \circ \gamma)'(t) \, dt d\sigma(\gamma) = \int_{A^I_\T} [\Phi(\gamma(b)) - \Phi(\gamma(a))] \, d\sigma(\gamma) = 0,
\end{align*}
where $\chi_\K$ denotes the characteristic function of the set $\K$.

\smallskip

\textbf{Causality $\eta$-a.e.:} Take any $\Phi \in C^\infty_{\cc}(\M_I^\circ)$, $\Phi \geq 0$ and any $\varphi \in L^2(\M_I, \eta)$ such that $\varphi \geq 0$ $\eta$-a.e. Let also $f \in C^\infty(\M_I^\circ)$ be a bounded causal function, which in particular means that $(f \circ \gamma)'(t) \geq 0$ for any $\gamma \in A^I_\T$ and $t \in I$ a.e. One can write that
\begin{align*}
\langle \wv{X}(\Phi f) - \wv{X}(\Phi)f, \varphi \rangle_{L^2(\M_I)} & = \int_{A^I_\T} \langle [(\Phi \circ \gamma)(f \circ \gamma)]' - (\Phi \circ \gamma)'(f \circ \gamma), \varphi \circ \gamma \rangle_{L^2(I)} \, d\sigma(\gamma)
\\
& = \int_{A^I_\T} \langle (\Phi \circ \gamma)(f \circ \gamma)', \varphi \circ \gamma \rangle_{L^2(I)} \, d\sigma(\gamma) 
\\
& = \int_{A^I_\T} \int_I \Phi(\gamma(t)) \varphi(\gamma(t)) (f \circ \gamma)'(t) \, dt d\sigma(\gamma) \geq 0
\end{align*}
by the positivity of $\Phi$, the positivity $\eta$-a.e. of $\varphi$, and the positivity a.e. of $(f \circ \gamma)'$.

We have thus shown that $\wv{X}$ defines an $L^2_\loc(\eta)$-vector field with all the properties that, on the strength of Remark \ref{Linfrem}, actually certify it is even $L^\infty_\loc(\eta)$-regular. All that now remains is to prove that the above properties alone (i.e. without resorting to $\sigma$ and formula (\ref{def_X0})) imply that the map $\Lambda_\mu(\Phi)$ belongs to $W^{1,\infty}(I)$ and that $\Lambda_\mu(\Phi)' = \Lambda_\mu(\wv{X}\Phi)$ for any $\Phi \in C^\infty_{\cc}(\M_I^\circ)$.

Both maps clearly belong to $L^\infty(I)$. In fact, $\Lambda_\mu$ can be regarded as a linear operator from $L^\infty(\M_I, \eta)$ to $L^\infty(I)$ defined via $\Lambda_\mu(F)(t) \vc \int_{\M_I} F d\mu_t$, which has norm one, because 
\begin{align*}
\|\Lambda_\mu(F)\|_{L^\infty(I)} = \esssup\limits_{t \in I} \left|\int_{\M_I} F d\mu_t \right| \leq \|F\|_{L^\infty(\M_I)} \esssup\limits_{t \in I} \underbrace{\left| \int_{\M_I} d\mu_t \right|}_{=\,1}.
\end{align*}

To see that $\Lambda_\mu(\wv{X}\Phi)$ is the weak derivative of $\Lambda_\mu(\Phi)$, take any test function $\phi \in C^\infty_\cc(I^\circ)$. Let us again fix $[a,b] \subset I^\circ$ such that $(a,b) \supset \T(\supp \Phi)$ and let $(\phi_n)$ be a sequence of polynomials\footnote{For example, one can use the Bernstein polynomials \cite{Bernstein}.} $C^1$-uniformly approximating $\phi$ on $[a,b]$. For our purposes it suffices to show that, for any $n$,
\begin{align}
\label{cont_eq2}
\int_a^b \Lambda_\mu(\Phi) \phi_n' dt = - \int_a^b \Lambda_\mu(\wv{X}\Phi) \phi_n dt.
\end{align}

To prove (\ref{cont_eq2}), notice first that the condition $\wv{X}(\Phi \T) - \wv{X}(\Phi)\T = \Phi$ immediately extends to $\wv{X}(\Phi (P \circ \T)) - \wv{X}(\Phi)(P \circ \T) = \Phi (P' \circ \T)$ for any polynomial $P$, and so one can write that
\begin{align*}
& \int_a^b \Lambda_\mu(\Phi)(t) \phi_n'(t) dt = \int_a^b \left( \int_{\M_I} \Phi d\mu_t \right) \phi_n'(t) dt = \int_a^b \int_{\M_I} \Phi (\phi_n' \circ \T) \, d\mu_t dt
\\
& = \underbrace{\int_a^b \int_{\M_I} \wv{X}(\Phi (\phi_n' \circ \T)) d\mu_t dt}_{= \, 0} - \int_a^b \int_{\M_I} \wv{X}\Phi \, (\phi_n \circ \T) d\mu_t dt
\\
& = - \int_a^b \left( \int_{\M_I} \wv{X}\Phi d\mu_t \right) \phi_n(t) dt = - \int_a^b \Lambda_\mu(\wv{X}\Phi)(t) \phi_n(t) dt,
\end{align*}
where $\int_a^b \int_{\M_I} \wv{X}(\Phi (\phi_n' \circ \T)) d\mu_t dt = \int_{\M_I} \wv{X}(\Phi (\phi_n' \circ \T)) d\eta = 0$ by the continuity equation. The integral running over $[a,b]$ can be extended onto entire $I$, because on $I \setminus [a,b]$ the integrand vanishes anyway, since
\begin{align*}
\supp \wv{X}(\Phi (\phi_n' \circ \T)) \subset \supp \Phi (\phi_n' \circ \T) \subset \supp \Phi \subset \T^{-1}((a,b)),
\end{align*}
where we have used (\ref{support_inclusion}). This concludes the proof of (iii).
\\

\emph{(iii) $\Rightarrow$ (i)} Fix $a,b \in I$ such that $a < b$. According to the characterization of causal relation between measures given by Proposition \ref{causalKchar}, proving that $\mu_a \preceq \mu_b$ amounts to showing that
\begin{align}
\label{3to1a}
\forall \K \subset \M, \K \textrm{ -- compact} \quad \mu_a(J^+(\K)) \leq \mu_b(J^+(\K)).
\end{align}
However, we claim it suffices to prove an analogous condition involving \emph{chronological} futures
\begin{align}
\label{3to1b}
\forall \K \subset \M, \K \textrm{ -- compact} \quad \mu_a(I^+(\K)) \leq \mu_b(I^+(\K)).
\end{align}
Indeed, fix a compact set $\K \subset \M$ and let $\overline{N_\varepsilon(\K)} \vc \{ p \in \M \, | \, d_h(p,\K) \leq \varepsilon \}$ be a closed $\varepsilon$-neighborhood of $\K$ with respect to some auxiliary complete Riemannian metric $h$. The set $\overline{N_\varepsilon(\K)}$ is itself compact (being closed and bounded, which by the Hopf--Rinov theorem is equivalent to being compact), and so in order to obtain (\ref{3to1a}) from (\ref{3to1b}) it is enough to prove that
\begin{align*}
J^+(\K) = \bigcap_{n=1}^\infty I^+(\overline{N_{1/n}(\K)}).
\end{align*}

To see ``$\subset$'', assume that $q \in J^+(\K)$, which means there exists $p \in \K$ such that $p \preceq q$. Notice now that for any $n \in \sN$ one can find $r \ll p$ such that $d_h(r,p) < 1/n$, which by (\ref{KPcausality}) implies $r \ll q$ and so $q \in I^+(\overline{N_{1/n}(\K)})$.

To see ``$\supset$'', assume that $q \in \bigcap_{n=1}^\infty I^+(\overline{N_{1/n}(\K)})$, which means that for any $n \in \sN$ there exists $p_n \in \overline{N_{1/n}(\K)}$ such that $p_n \ll q$ and so also $p_n \preceq q$. The sequence $(p_n)$ is contained in the compact set $\overline{N_1(\K)}$, and so it has a subsequence convergent to some $p \in \bigcap_{n=1}^\infty \overline{N_{1/n}(\K)} = \K$. By the topological closedness of $J^+$, we obtain $p \preceq q$ and so $q \in J^+(\K)$.

We now move to proving (\ref{3to1b}). To this end, fix the compact set $\K \subset \M$ along with $c > b$ and define $\cS \vc \partial I^+(\M_c \cup \K)$. We claim that $\cS$ is a Cauchy hypersurface (possibly nonsmooth and nonspacelike) satisfying, additionally,
\begin{align}
\label{3to1c}
I^+(\cS) = I^+(\M_c \cup \K),
\end{align}
which in turn means that 
\begin{align}
\label{3to1d}
\forall \, t \in [a,b] \qquad \M_t \cap I^+(\K) = \M_t \cap I^+(\cS).
\end{align}
Indeed, since $I^+(\M_c \cup \K)$ is a future set, its boundary $\cS$ is achronal \cite[Chapter 14, Corollary 27]{BN83}. Let $\gamma: \sR \rightarrow \M$ be any inextendible timelike curve parametrized such that $\T \circ \gamma = \textnormal{id}_\sR$. Notice that $\gamma(t) \in I^+(\M_c \cup \K)$ for $t > c$, whereas $\gamma(t) \not\in I^+(\M_c \cup \K)$ for $t < \min \T(\K) \cup \{c\}$. Therefore, the curve $\gamma$ must cross the boundary $\cS = \partial I^+(\M_c \cup \K)$ and it does so exactly once by the latter's achronality. This proves that $\cS$ is a Cauchy hypersurface.

In order to obtain (\ref{3to1c}), we prove the following lemma.

\begin{lem}
\label{lem_3to1}
Let $\Y \subset \M$ be such that $\Y \subset I^+(\X)$ for some achronal set $\X \subset \M$. Then $I^+(\partial \Y) = I^+(\Y)$. 
\end{lem}
\begin{proof}
Assume $q \in I^+(\partial \Y)$, i.e. there exists $p \in \partial \Y$ such that $p \ll q$. Since $I^-(q)$ is a neighborhood of $p$, one can find $r \in \Y \cap I^-(q)$. But this means that $q \in I^+(\Y)$.

Conversely, assume $q \in I^+(\Y)$, i.e. there exists $p \in \Y$ such that $p \ll q$. By assumption $\exists x \in \X \quad x \ll p$. Of course $x \not\in \Y$ by the achronality of $\X$, but this means that the timelike curve joining $x$ with $p$ must cross $\partial \Y$ at some point and so $p \in I^+(\partial \Y)$. This, in turn, implies that $q \in I^+(\partial \Y)$ as well.
\end{proof}

Equality (\ref{3to1c}) follows from taking $\Y \vc I^+(\M_c \cup \K)$ in the above lemma. Notice that for the achronal set $\X$ one might take any Cauchy hypersurface $\M_t$ with $t < \min \T(\K) \cup \{c\}$.

We now invoke one of the results of Bernal \& S\'anchez \cite[Theorem 5.15]{BS06}, by which there exists a smooth surjective causal\footnote{Bernal \& S\'anchez do not call their function ``causal'', but it directly follows from the list of its properties.} function $f: \M \rightarrow \sR$ such that $\cS = f^{-1}(0)$. Let $(\varphi_n) \subset C^\infty(\sR)$ be a sequence of nondecreasing functions such that, for any $n \in \sN$, $0 \leq \varphi_n \leq 1$, $\varphi_n|_{(-\infty,0]} = 0$, and $\varphi_n \rightarrow \chi_{(0,+\infty)}$ pointwise. Concretely, one may take
\begin{align*}
\varphi_n(x) \vc \left\{ \begin{array}{ll} \exp(-1/nx) & \textnormal{for } x > 0 \\ 0 & \textnormal{for } x \leq 0 \end{array} \right.
\end{align*}
Observe that $(\varphi_n \circ f)$ is a sequence of smooth bounded causal functions vanishing on $J^-(\cS)$, positive on $I^+(\cS)$ and such that $\varphi_n \circ f \rightarrow \chi_{I^+(\cS)}$ pointwise.

Finally, as the last ingredients we need, let $\{\K_m\}$ be an exhaustion of $\M$ by compact sets (with $\K_m \subset \K_{m+1}^\circ$) and let $(\Omega_m)$ be a sequence of smooth compactly supported functions such that $0 \leq \Omega_m \leq 1$, $\supp \Omega_m \subset \K^\circ_{m+1}$, and $\Omega_m|_{\K_m} = 1$ for every $m \in \sN$. Needless to say, $\Omega_m \rightarrow 1$ pointwise.

On the strength of (iii), we can now write that
\begin{align*}
& \int_{\M_I} (\varphi_n \circ f) d\mu_b - \int_{\M_I} (\varphi_n \circ f) d\mu_a = \lim\limits_{m \rightarrow \infty} \left( \int_{\M_I} \Omega_m (\varphi_n \circ f) d\mu_b - \int_{\M_I} \Omega_m (\varphi_n \circ f) d\mu_a \right)
\\
& = \lim\limits_{m \rightarrow \infty} \int_a^b \left( \int_{\M_I} \Omega_m (\varphi_n \circ f) d\mu_t \right)' dt = \lim\limits_{m \rightarrow \infty} \int_a^b \int_{\M_I} \wv{X}(\Omega_m (\varphi_n \circ f)) d\mu_t dt 
\\
& \geq \lim\limits_{m \rightarrow \infty} \int_a^b \int_{\M_I} \wv{X}(\Omega_m) (\varphi_n \circ f) d\mu_t dt = \lim\limits_{m \rightarrow \infty} \int_a^b \int_{\M_t} \wv{X}(\Omega_m) (\varphi_n \circ f) d\mu_t dt,
\end{align*}
where the limit can be taken out of the integral on the strength of the monotone convergence theorem, the inequality is the direct application of ``causality $\eta$-a.e.'' condition in (iii) and the last equality simply follows from $\supp \mu_t \subset \M_t$.

We now make the crucial claim that the integral on the right-hand side of the above inequality becomes zero for sufficiently large $m$ and hence the entire limit is zero. The argument is as follows.

Firstly, observe that because $\varphi_n \circ f$ is nonvanishing only on $I^+(\cS)$, the inner integral runs over $\M_t \cap I^+(\cS)$ for any $t \in [a,b]$, which by (\ref{3to1d}) is the same as $\M_t \cap I^+(\K)$. This, in turn, means that the double integral runs over $I^+(\K) \cap \M_{[a,b]} \subset J^+(\K) \cap J^+(\M_a) \cap J^-(\M_b)$, the latter set being compact (cf. \cite[Proposition 13 (iii)]{Miller17a}). Hence, there exists $m_0$ such that $I^+(\K) \cap \M_{[a,b]} \subset \K_m^\circ$ for all $m \geq m_0$. For every such $m$ the function $\Omega_m$ is identically equal to $1$ on $\K_m$, and hence, in particular, on the entire area of integration. But this already implies that $\wv{X}(\Omega_m)|_{\K_m^\circ} = 0$ $\eta$-a.e. Indeed, if $X$ is the extension of $\wv{X}$ as given by Theorem \ref{thm_VF}, then by locality (Corollary \ref{cor_VF}) we can write (cf. also (\ref{thm_VF5}))
\begin{align*}
\wv{X}(\Omega_m)|_{\K_m^\circ} = X(\Omega_m)|_{\K_m^\circ} = X(1)|_{\K_m^\circ} = 0.
\end{align*}


All in all, we have thus obtained that for any $n \in \sN$
\begin{align*}
\int_{\M_I} (\varphi_n \circ f) d\mu_a \leq \int_{\M_I} (\varphi_n \circ f) d\mu_b.
\end{align*}

Passing now to the limit $n \rightarrow +\infty$, by the dominated convergence theorem we get
\begin{align*}
\int_{\M_I} \chi_{I^+(\cS)} d\mu_a \leq \int_{\M_I} \chi_{I^+(\cS)} d\mu_b
\end{align*}
or, equivalently, $\mu_a(I^+(\cS)) \leq \mu_b(I^+(\cS))$. But this means that
\begin{align*}
\mu_a(I^+(\K)) & = \mu_a(\M_a \cap I^+(\K)) = \mu_a(\M_a \cap I^+(\cS)) = \mu_a(I^+(\cS))
\\
& \leq \mu_b(I^+(\cS)) = \mu_b(\M_b \cap I^+(\cS)) = \mu_b(\M_b \cap I^+(\K)) = \mu_b(I^+(\K)),
\end{align*}
where we used (\ref{3to1d}). This finishes the proof of (\ref{3to1b}) and thus the entire proof of (i).
\end{proof}

\section{Transformation laws and invariants}
\label{sec::discussion}

Theorem \ref{premain} (reexpressed as Theorem \ref{main}), explicitly involves the chosen Cauchy temporal function $\T$ and its associated GBS-splitting. In other words, it pertains to a concrete `global observer', who describes a causal evolution of some spatially spread physical quantity $Q$ (modeled by a time-dependent probability measure), employing his or her notion of time and space. Were we to choose a different Cauchy temporal function, we would have a different trio of equivalent descriptions (i--iii). One might naturally expect, however, that there is way to transform between various observers and their descriptions, and that there exist splitting-independent geometrical objects behind those descriptions. Here we find two such objects: the probability measure on the space of unparametrized causal curves (worldlines), and the `probability 4-current'.

To simplify, we restrict ourselves to the case where $I = \sR$. Suppose that the observers $\bA$ and $\bB$, who employ Cauchy temporal functions $\T_\bA$, $\T_\bB$, respectively, model the causal evolution of the quantity $Q$ using the maps $t \mapsto \mu_t$ and $\tau \mapsto \nu_\tau$, which satisfy $\supp \mu_t \subset \T_\bA^{-1}(t)$ and $\supp \nu_\tau \subset \T_\B^{-1}(\tau)$ for all $t,\tau \in \sR$. Denote $\eta_\bA \vc \int_\sR \mu_t dt$ and $\eta_\bB \vc \int_\sR \nu_\tau d\tau$.

Consider the spaces $A_{\T_\bA}^\sR$ and $A_{\T_\bB}^\sR$ of causal curves as parametrized by observer $\bA$ and $\bB$, respectively. The simple reparametrization map allows to pass between these spaces in a continuous manner.
\begin{prop}
\label{reparametrization}
The map $\widetilde{\ }: A^\sR_{\T_\bA} \rightarrow A^\sR_{\T_\bB}$ defined via $\widetilde{\gamma} \vc \gamma \circ (\T_\bB \circ \gamma)^{-1}$ is a well-defined reparametrization of causal curves and a homeomorphism.
\end{prop}
\begin{proof}
Cf. \cite[Proposition 9]{Miller17a}.
\end{proof}

Consequently, when lifted onto the level of probability measures, the reparametrization map $\widetilde{\ }$ translates \textbf{A}'s description (ii) into \textbf{B}'s. Concretely, if $\sigma \in \Pf(A^\sR_{\T_\bA})$ satisfies $(\ev_t)_\sharp\sigma = \mu_t$, then the pushforward measure $\widetilde{\sigma} \vc \widetilde{\;}_\sharp \sigma \in \Pf(A^\sR_{\T_\bB})$ should satisfy $(\ev_\tau)_\sharp\widetilde{\sigma} = \nu_\tau$, yielding a (rather implicit) relationship between the evolutions $t \mapsto \mu_t$ and $\tau \mapsto \nu_\tau$.

However, Proposition \ref{reparametrization} has another important consequence. Let $\Cf_\textrm{inext}$ be the set of all inextendible, unparametrized causal curves in $\M$. The map $\im: A_{\T_\bA}^\sR \rightarrow \Cf_\textrm{inext}$, which assigns to every curve $\gamma$ its image $\im(\gamma) \vc \gamma(\sR)$ is a well-defined bijection \cite[Proposition 8]{Miller17a}. If we endow $\Cf_\textrm{inext}$ with the (locally compact Polish space) topology transported from $A_{\T_\bA}^\sR$ by the map $\im$, then Proposition \ref{reparametrization} guarantees that this topology in fact does not depend on the choice of the Cauchy temporal function. In this way, the pushforward measure $\upsilon \vc \im_\sharp \sigma \in \Pf(\Cf_\textrm{inext})$ is the desired splitting-independent object behind description (ii).

We now move our attention to description (iii). Observe, crucially, that the right-hand side of formula (\ref{def_X0}) --- when evaluated for smooth compactly supported test function $\varphi$ --- is observer-independent.
\begin{prop}
For any $\sigma \in \Pf(A^\sR_{\T_\bA})$ and $\Phi, \varphi \in C^\infty_{\cc}(\M)$
\begin{align}
\label{obsind_X0}
\int_{A^\sR_{\T_\bA}} \int_\sR (\Phi \circ \gamma)'(t) (\varphi \circ \gamma)(t) dt \, d\sigma(\gamma)
 = \int_{A^\sR_{\T_\bB}} \int_\sR (\Phi \circ \widetilde{\gamma})'(\tau) (\varphi \circ \widetilde{\gamma})(\tau) d\tau \, d\widetilde{\sigma}(\widetilde{\gamma}).
\end{align}
\end{prop}
\begin{proof}
By the proof of ``(ii) $\Rightarrow$ (iii)'' of Theorem \ref{main} we know that the map $A^\sR_{\T_\bB} \ni \widetilde{\gamma} \mapsto \langle (\Phi \circ \widetilde{\gamma})', \varphi \circ \widetilde{\gamma} \rangle_{L^2(\sR)}$ is $\widetilde{\sigma}$-integrable. Thus, we are allowed to perform the change of variables:
\begin{align*}
\int_{A^\sR_{\T_\bB}} \langle (\Phi \circ \widetilde{\gamma})', \varphi \circ \widetilde{\gamma} \rangle_{L^2(\sR)} d\widetilde{\sigma}(\widetilde{\gamma}) = \int_{A^\sR_{\T_\bA}} \langle (\Phi \circ \gamma \circ (\T_\bB \circ \gamma)^{-1})', \varphi \circ \gamma \circ (\T_\bB \circ \gamma)^{-1} \rangle_{L^2(\sR)} d\sigma(\gamma).
\end{align*}
Substituting now $t = (\T_\bB \circ \gamma)^{-1}(\tau)$ in the inner integral, we obtain, for any $\gamma \in A^\sR_{\T_\bA}$,
\begin{align*}
& \langle (\Phi \circ \gamma \circ (\T_\bB \circ \gamma)^{-1})', \varphi \circ \gamma \circ (\T_\bB \circ \gamma)^{-1} \rangle_{L^2(\sR)}
\\
& = \int_\sR [\Phi \circ \gamma \circ (\T_\bB \circ \gamma)^{-1}]'(\tau) \, [\varphi \circ \gamma \circ (\T_\bB \circ \gamma)^{-1}](\tau) \, d\tau \\
& = \int_\sR (\Phi \circ \gamma)'(t) \, \tfrac{1}{(\T_\bB \circ \gamma)'(t)} \, (\varphi \circ \gamma)(t) \, (\T_\bB \circ \gamma)'(t) dt
\\
& = \int_\sR (\Phi \circ \gamma)'(t) \, (\varphi \circ \gamma)(t) dt = \langle (\Phi \circ \gamma)', \varphi \circ \gamma \rangle_{L^2(\sR)},
\end{align*}
what concludes the proof.
\end{proof}

Let $\wv{X}_\bA$, $\wv{X}_\bB$ be the $L^2_\loc$-vector fields constructed by observers \textbf{A} and \textbf{B} from measures $\sigma$ and $\widetilde{\sigma}$, respectively. Formula (\ref{obsind_X0}) can be expressed simply as
\begin{align}
\label{obsind_X1}
\int_\M \wv{X}_\bA \Phi \, \varphi \, d\eta_\bA = \int_\M \wv{X}_\bB \Phi \, \varphi \, d\eta_\bB
\end{align}
for all $\Phi, \varphi \in C^\infty_{\cc}(\M)$. Using Theorem \ref{thm_VF} and the standard density argument, one can easily extend this equality to
\begin{align}
\label{obsind_X2}
\int_\M X_\bA \Psi \, \varphi \, d\eta_\bA = \int_\M X_\bB \Psi \, \varphi \, d\eta_\bB,
\end{align}
where this time $\Psi \in C^\infty(\M)$ and $\varphi \in C_{\cc}(\M)$. We thus arrive at another splitting-independent object --- this time behind description (iii). It is the map $\eta X$ assigning to each $\Psi \in C^\infty(\M)$ a linear functional on $C_{\cc}(\M)$ defined via any side of formula (\ref{obsind_X2}). This functional is in fact a \emph{real-valued Radon measure}, i.e. a difference of two (positive) Radon measures\footnote{One might also adopt the Bourbaki's approach \cite{Bourbaki} (cf. also \cite[Chapter XIII]{Dieudonne}) and \emph{define} the real-valued Radon measure as a continuous linear functional on $C_\cc(\M)$ endowed with a suitable locally convex topology.} $(X\Psi)^+\eta - (X\Psi)^-\eta$, with $(X\Psi)^\pm \in L^{\infty}_\loc(\M,\eta)$ denoting the positive/negative parts of $X\Psi$. By equation (\ref{conteq3}) and the subsequent discussion, the map $\eta X$ deserves to be called the \emph{Radon 4-current}, being a low-regular generalization of the smooth 4-current $\rho Y$. 

By plugging $\Psi \vc \T_\bB$ into formula (\ref{obsind_X2}), we obtain the transformation rules for the measures $\eta_{\bA,\bB}$ and, consequently, also for the vector fields $X_{\bA,\bB}$. Namely,
\begin{align}
\label{obsind_X3}
\eta_\bB = X_\bA \T_\bB \, \eta_\bA, \qquad \textrm{whereas} \qquad X_\bB = \tfrac{1}{X_\bA \T_\bB} X_\bA,
\end{align}
where Proposition \ref{causal_measure_map} guarantees that everything is well defined. The second formula is quite expected --- observer \textbf{B} simply rescales the vector field $X_\bA$ so as to have $X_\bB \T_\bB = 1$. As for the first formula, by disintegrating the $\eta$-measures with respect to $\T_\bB$, we can write another relationship between the evolutions $t \mapsto \mu_t$ and $\tau \mapsto \nu_\tau$, which reads
\begin{align}
\label{obsind_X4}
\nu_\tau = X_\bA \T_\bB \, \eta_\bA^\tau \qquad \textnormal{for almost all } \tau,
\end{align}
where $\{\eta_\bA^\tau\}_{\tau \in \sR}$ is the disintegration\footnote{See \cite[452E \& 452O]{Fremlin4} for the general definition of disintegration (not limited to probability measures) and the existence theorem needed in the present context.} of $\eta_\bA$ with respect to $\T_\bB$.

To finish this section, let us write down the transformation rules for the vector fields $Y_{\bA,\bB} \vc |\nabla\T_{\bA,\bB}| X_{\bA,\bB}$, which, as explained in Sec. \ref{sec::intro}, have the interpretation of 4-velocities of the probability flux as seen by observers $\bA,\bB$, respectively. By the second formula of (\ref{obsind_X3}), we obtain that
\begin{align*}
Y_\bB = \tfrac{|\nabla \T_\bB|}{Y_\bA \T_\bB} Y_\bA.
\end{align*}
If, additionally, any (and hence both) of the measures $\eta_{\bA,\bB}$ is absolutely continuous with respect to the volume measure $\vol_g$ with density $\rho_{\bA,\bB}$, respectively, then we also have that
\begin{align*}
\rho_\bB = \tfrac{Y_\bA \T_\bB}{|\nabla \T_\bB|} \rho_\bA
\end{align*}
and so $\rho_\bA Y_\bA = \rho_\bB Y_\bB$, as one should expect.

\section*{Appendix A: Notational remarks}
\label{sec::AppA}

\renewcommand{\theequation}{A.\arabic{equation}}
\setcounter{equation}{0}
\renewcommand{\thethm}{A.\arabic{thm}}
\setcounter{thm}{0}

Throughout the paper, for any subset $\X$ of a topological space $\M$ we denote its complement, closure, interior and boundary by $\X^c$, $\overline{\X}$, $\X^\circ$ and $\partial\X$, respectively. If $\Y$ is another subset of $\M$, we say that $\X$ is compactly embedded in $\Y$, what is written symbolically as $\X \Subset \Y$, if $\overline{\X}$ is compact and contained in $\Y^\circ$. Notice that for a bounded interval $I$, $(a,b) \Subset I$ amounts to $[a,b] \subset I^\circ$.

The notion of a \emph{support} is used in three closely related, but formally different meanings.

The most basic meaning concerns functions. If $f: \M \rightarrow \sR^N$ is a function on a topological space $(\M, \tau)$, then $\supp f \vc \overline{\{x \in \M \, | \, f(x) \neq 0\}}$. It is worthwhile to observe that the complement of the support $(\supp f)^c$ is the largest open set on which $f$ vanishes. In other words, $\supp f = (\bigcup_{\V \in \tau, \, f|_\V = 0} \V )^c$.

Defining the support through its complement has the advantage of allowing for a straightforward extension onto classes of functions equal $\nu$-a.e. for some fixed Borel measure $\nu$ as well as onto Borel measures themselves. Concretely, let $u$ be a class of $\sR^N$-valued functions on $\M$ equal $\nu$-a.e. One defines $\supp u = (\bigcup_{\V \in \tau, \, u|_\V = 0 \, \nu\textnormal{-a.e}} \V )^c$. In the same vein, the support of the Borel measure $\nu$ on $\M$ is defined via $\supp u = (\bigcup_{\V \in \tau,\, \nu(\V) = 0} \V )^c$.

Even though we employ the same symbol `$\supp$' in all three cases, it is always clear from the context which one of them applies.

The space of continuous maps between topological spaces $\M,\N$ is denoted by $C(\M,\N)$. If only the compactly supported maps are considered, the respective space is denoted by $C_\cc(\M,\N)$. If $\N = \sR$, the spaces of continuous, continuous bounded, continuous and compactly supported real-valued maps are denoted simply by $C(\M), C_\cb(\M), C_\cc(\M)$, respectively. If $\M$ is a manifold, then $C^\infty(\M,\sR^N)$ ($C^\infty_\cc(\M,\sR^N)$) denotes the space of smooth (compactly supported) $\sR^N$-valued maps. If $N=1$, again we simply write $C^\infty(\M)$ ($C^\infty_\cc(\M)$).

All measures appearing in the paper are \emph{Borel} measures, i.e. measures defined on the $\sigma$-algebra generated by the topology of a given topological space. In fact, they are even \emph{Radon} measures, i.e. regural Borel measures which are finite on all compact sets. Given a  Borel map $F: \M \rightarrow \N$ between topological spaces and a (Borel) measure $\nu$ on $\M$, we denote the pushforward measure $\nu \circ F^{-1}$ by $F_\sharp$.

\section*{Appendix B: Causality theory}
\label{sec::AppB}

\renewcommand{\theequation}{B.\arabic{equation}}
\setcounter{equation}{0}
\renewcommand{\thethm}{B.\arabic{thm}}
\setcounter{thm}{0}

In order to make the paper self-contained, we recollect some definitions and basic results from causality theory. For a full exposition, consult e.g. \cite{Beem,BN83,HawkingEllis,MS08,Penrose1972}.

Let $(\M,g)$ be a spacetime, i.e. a connected time-oriented smooth Lo\-rentz\-ian manifold\footnote{Recall that we employ the signature convention $(-++\ldots+)$.}. The Lorentzian metric $g$ induces on $\M$ binary relations $\ll$ and $\preceq$, called the \emph{chronological} and the \emph{causal precedence relations}, respectively. We say that $p$ chronologically (reps. causally) precedes $q$, or that $q$ is in the chronological (resp. causal) future of $p$, which is denoted $p \ll q$ (resp. $p \preceq q$), if there exists a piecewise smooth future-directed chronological (resp. causal) curve from $p$ to $q$. By the standard convention, we also assume $p \preceq p$. By $p \prec q$ we mean that $p \preceq q$, but $p \neq q$.

Both $\ll$ and $\preceq$ are transitive relations and, additionally,
\begin{align}
\label{KPcausality}
\forall p,q,r \in \M \qquad (p \ll q \preceq r) \vee (p \preceq q \ll r) \quad \Rightarrow \quad p \ll r.
\end{align}

With the help of $\preceq$, one extends the notion of a causal curve to the curves which are only $C^0$. Although the general definition of a continuous (future-directed) causal curve is somewhat convoluted (cf. \cite[Definition 3.15]{MS08}), it simplifies greatly for the so-called distinguishing spacetimes. Namely \cite[Proposition 3.19]{MS08}, the map $\gamma \in C(I,\M)$ is called future-directed causal if
\begin{align*}
\forall \, s,t \in I \qquad s < t \ \Rightarrow \ \gamma(s) \prec \gamma(t).
\end{align*}

Notice that the above condition, even though it seems stronger than the one used in Definition \ref{AIT_def}, is nevertheless equivalent to (\ref{AIT_def1}) since the elements of $A_\T^I$ are by definition injective maps.

The relations $\ll$ and $\preceq$ are formally subsets of $\M^2$, but for historical reasons one usually denotes these subsets as $I^+ \vc \{ (p,q) \in \M^2 \ | \ p \ll q \}$ (which is open) and $J^+ \vc \{ (p,q) \in \M^2 \ | \ p \preceq q \}$ (which is $\sigma$-compact \cite[Theorem 4]{AHP2017}). Moreover, one writes $I^\pm(p)$ (resp. $J^\pm(p)$) to denote the set of all events in the chronological (resp. causal) future/past of $p$. Finally, for any $\X \subset \M$ one introduces $I^\pm(\X) \vc \bigcup_{p \in \X} I^\pm(p)$ and similarly $J^\pm(\X) \vc \bigcup_{p \in \X} J^\pm(p)$.

A~causal curve $\gamma \in C((a,b),\M)$, where $-\infty \leq a < b \leq +\infty$ is called \emph{inextendible} if both the limits $\lim_{t \rightarrow a^+} \gamma(t)$ and $\lim_{t \rightarrow b^-} \gamma(t)$ do not exist. A~\emph{Cauchy hypersurface} is a~subset $\cS \subseteq \M$ met by every inextendible timelike curve exactly once. Any such $\cS$ is a~closed achronal topological hypersurface, met by every inextendible causal curve. Obviously, if $\cS$ is additionally spacelike, then the latter sentence can be appended with the phrase ``exactly once'' as well.

A~function $\T: \M \rightarrow \mathbb{R}$ is referred to as
\begin{itemize}
\item a~\emph{causal function} iff it is nondecreasing along every future-directed causal curve;
\item a~\emph{time function} iff it is continuous and strictly increasing along every future-directed causal curve;
\item a~\emph{temporal function} iff it is a~smooth function with a past-directed timelike gradient.
\end{itemize}

Each of the above definitions is stronger than the preceding one. Mind that even a smooth time function need not be temporal. Additionally, any time (temporal) function the level sets of which are Cauchy hypersurfaces is called a \emph{Cauchy} time (temporal) function.

Spacetimes can be classified according to their increasingly better causal features. Each of the levels of this ``causal hierarchy'' can be characterized in many equivalent ways. Here we recall only those notions which are needed in the current paper --- the more complete exposition can be found e.g. in \cite{MS08}. 

A spacetime $\M$ is called
\begin{itemize}
\item \emph{causal} if it does not admit closed causal curves. This is equivalent to $\preceq$ being an antisymmetric relation (and thus a partial order);
\item \emph{distinguishing} if $I^\pm(p) = I^\pm(q)$ implies $p = q$ for any $p,q \in \M$;
\item \emph{stably causal} if it admits time functions or, equivalently, if it admits temporal functions \cite{BS04};
\item \emph{causally simple} if it is causal and the~sets $J^\pm(p)$ are closed for every $p \in \M$;
\item \emph{globally hyperbolic} if it is causal and the~sets $J^+(p) \cap J^-(q)$ are compact for all $p,q \in \M$.
\end{itemize}

In stably causal spacetimes, one can use temporal as well as smooth bounded causal functions to characterize the future-directed causal vectors.
\begin{lem}
\label{lem_CT}
Let $(\M,g)$ be a stably causal spacetime, and let $v \in T_q\M$, $v \neq 0$ for some fixed $q \in \M$. The following conditions are equivalent:
\begin{enumerate}[i)]
\item $v$ is future-directed causal.
\item For any temporal function $\T \quad v(\T) > 0$.
\item For any smooth bounded causal function $f \quad v(f) \geq 0$.
\end{enumerate}
\end{lem}
\begin{proof}
$(i) \Rightarrow (ii)$ The gradient $\nabla\T_q$ is by definition past-directed timelike, and so one has $v(\T) = g(v, \nabla\T_q) > 0$.
\\

$(ii) \Rightarrow (iii)$ Since $f$ is a smooth causal function, its gradient at $q$ is past-directed causal or zero \cite[Proposition 1.18]{MinguzziNEW} and hence $\nabla f_q + \varepsilon \nabla\T_q$ is past-directed timelike for any temporal function $\T$, any $\varepsilon > 0$ and any $q \in \M$. In other words, $f + \varepsilon \T$ is a temporal function for every $\varepsilon > 0$. On the strength of \emph{(ii)}, we get $v(f) + \varepsilon v(\T) = v(f + \varepsilon \T) > 0$, and so passing with $\varepsilon \rightarrow 0^+$ yields \emph{(iii)}.
\\

$(iii) \Rightarrow (i)$ Suppose $v$ is past-directed causal or spacelike. In the former case, for any temporal (and hence smooth and causal) function $\T$ we would get $v(\arctan \circ \T) < 0$ (where composition with $\arctan$ guarantees boundedness), contradicting $(iii)$. In the latter case, let $\widetilde{\T}$ be a fixed temporal function and let $\varepsilon > 0$ be small enough for $w \vc v - \varepsilon \nabla\widetilde{\T}_q$ be spacelike (recall that the set of nonzero spacelike vectors is open in the standard linear topology of $T_q\M$). We claim there exists a bounded temporal function $\T$ such that $w(\T) = 0$. Indeed, let $i$ be a smooth conformal embedding of $\M$ into an $N$-dimensional Minkowski spacetime $\mathbb{L}^N$ existing by \cite[Corollary 1.4]{LNembeddings}. Conformality implies that the transported vector $di_q(w)$ is spacelike, and so one can choose a time-coordinate $x^0$ in $\mathbb{L}^N$ such that $di_q(w)(x^0) = 0$. Conformality guarantees, moreover, that the function $x^0 \circ i$ is temporal, and so the bounded temporal function $\T \vc \arctan \circ \, x^0 \circ i$ satisfies
\begin{align*}
w(\T) = \tfrac{1}{1 + (x^0(i(q)))^2} \, w(x^0 \circ i) = \tfrac{1}{1 + (x^0(i(q)))^2} \, di_q(w)(x^0) = 0,
\end{align*}
which in turn yields
\begin{align*}
v(\T) =  (w + \varepsilon \nabla\widetilde{\T}_q)(\T) = \varepsilon g(\nabla\widetilde{\T}_q, \nabla\T_q) < 0,
\end{align*}
again contradicting $(iii)$.
\end{proof}

\begin{rem}
The above lemma remains true if we replace `bounded causal' in condition \emph{(iii)} with `time', `bounded time' or `causal'. Indeed, each of thus modified conditions implies the original one and is itself implied by \emph{(ii)}.
\end{rem}

\begin{prop}
\label{PropCS}
If $(\M,g)$ is a causally simple spacetime, then $J^+$ is a closed subset of $\M^2$, whereas $J^+(\K)$ and $J^-(\K)$ are closed subsets of $\M$ for every compact $\K \subset \M$.
\end{prop}
\begin{proof}
See \cite[Lemma 3.67]{MS08}.
\end{proof}

In causally simple spacetimes the relation $\preceq$ between measures (Definition \ref{causal_rel_def}) allows for various equivalent definitions \cite{AHP2017}, out of which in the current paper we need the following one.
\begin{prop}
\label{causalKchar}
If $(\M,g)$ is a causally simple spacetime, then for any $\nu_1,\nu_2 \in \Pf(\M)$
\begin{align*}
\nu_1 \preceq \nu_2 \qquad \Longleftrightarrow \qquad \forall \K \subset \M, \ \K \textnormal{ -- compact} \quad \nu_1(J^+(\K)) \leq \nu_2(J^+(\K)).
\end{align*}
\end{prop}
\begin{proof}
See \cite[Theorem 8.]{AHP2017} or \cite[Theorem 4.]{Miller18} for an alternative proof.
\end{proof}

For the next result, recall that for any Polish space $\X$ the space $\Pf(\X)$ of all Borel probability measures on $\X$ is itself Polish when endowed with the \emph{narrow topology}\footnote{Also referred to as the $w^\ast$-topology \cite{Aliprantis} or the $w$-topology \cite{garling2017polish}.}, defined as the coarsest topology such that the maps $\Pf(\X) \ni \nu \mapsto \int_\X \phi \, d\nu$ are continuous for all $\phi \in C_\cb(\X)$.
\begin{prop}
\label{narrowcom}
If $(\M,g)$ is a causally simple spacetime, then for any $\nu_1,\nu_2 \in \Pf(\M)$ the set $\itPi_\preceq(\nu_1, \nu_2)$ of causal couplings of $\nu_1$ and $\nu_2$ is narrowly compact in $\Pf(\M^2)$.
\end{prop}
\begin{proof}
See \cite[Lemma 8]{Miller17a}.
\end{proof}

Finally, in globally hyperbolic spacetimes many more intersections of causal futures and pasts --- not just those of the form $J^+(p) \cap J^-(q)$, $p,q \in \M$ --- turn out to be compact.
\begin{prop}
\label{PropCCC}
$(\M,g)$ is a globally hyperbolic spacetime iff it admits Cauchy hypersurfaces. Moreover, the following subsets of $\M$ are compact:
\begin{itemize}
\item $J^+(\K_1) \cap J^-(\K_2)$ for any compact $\K_1, \K_2 \subseteq \M$,
\item $J^\pm(\K) \cap \cS$ for any compact $\K \subseteq \M$ and Cauchy hypersurface $\cS$,
\item $J^\pm(\K) \cap J^\mp(\cS)$ for any compact $\K \subseteq \M$ and Cauchy hypersurface $\cS$.
\end{itemize}
\end{prop}
\begin{proof}
See \cite[Proposition 13]{Miller17a}.
\end{proof}

\section*{Appendix C: Sobolev spaces of univariate functions}
\label{sec::AppC}

\renewcommand{\theequation}{C.\arabic{equation}}
\setcounter{equation}{0}
\renewcommand{\thethm}{C.\arabic{thm}}
\setcounter{thm}{0}

For Reader's convenience, below we include a more detailed exposition of the theory of Sobolev spaces $W^{k,p}(I,\sR^N)$ and $W^{k,p}_{\loc}(I,\sR^N)$, where $I \subset \sR$ is an interval, needed throughout the paper. The latter spaces are examples of \emph{local} Sobolev spaces, which  in the textbooks (i.a. \cite{Adams,Evans,Mazya,Ziemer}) typically play an auxiliary role and which have been systematically studied only recently \cite{LocSob}. From the point of view of the current paper, what we are mostly interested in are certain \emph{topological} properties of the spaces $L^2_{\loc}(I,\sR^N)$ and $W^{1,2}_{\loc}(I,\sR^N)$. Nevertheless, let us begin by recalling the general definition (cf. \cite{WeberSobolev}). 

Let $\cdot$ and $| . |$ denote the standard Euclidean inner product and norm, respectively, and let $\cL$ stand for the 1-dimensional Lebesgue measure.
\begin{defn}
\label{locint}
Let $p \in [1,\infty)$. An $\cL$-measurable function $u: I \rightarrow \sR^N$ is
\begin{itemize}
\item $p$-\emph{integrable} if $\int_I |u(t)|^p dt < \infty$.
\item \emph{essentially bounded} if $\exists c \in \sR \ \cL(\{t \in I \, | \, |u(t)| > c\}) = 0$.
\item $p$-\emph{locally integrable} if $\forall (a,b) \Subset I \ \int_a^b |u(t)|^p dt < \infty$.
\item \emph{locally essentially bounded} if \\ $\forall (a,b) \Subset I \ \exists c \in \sR \ \cL(\{t \in (a,b) \, | \, |u(t)| > c\}) = 0$.
\item said to admit a \emph{weak derivative} of order $k > 0$ if there exists a $1$-locally integrable\footnote{1-local integrability is the least restrictive property among those discussed in Definition \ref{locint}. It is implied by the $p$-local integrability for any other $p > 1$, which is in turn implied by the local essential boundedness. Of course, $p$-integrability and essential boundedness imply their local versions. In other words, the $L^1_\loc$-regularity assumption on the weak derivative is the weakest possible (among Sobolev spaces).} function $u^{(k)}: I \rightarrow \sR^N$ such that
\begin{align*}
\forall \phi \in C^\infty_\cc(I^\circ,\sR^N) \quad \int_I u \cdot \phi^{(k)} dt = (-1)^k \int_I u^{(k)} \cdot \phi \, dt,
\end{align*}
where we have suppressed the functions' arguments for legibility. The weak derivative $u^{(k)}$, if it exists, is unique a.e.
\end{itemize}

The spaces of classes of functions equal a.e. satisfying each of the first four conditions are denoted $L^p(I,\sR^N)$, $L^\infty(I,\sR^N)$, $L^p_\loc(I,\sR^N)$ and $L^\infty_\loc(I,\sR^N)$, respectively. One subsequently defines the \emph{Sobolev spaces}
\begin{align*}
& W^{k,p}(I,\sR^N) = \{ [u] \ | \ [u], [u^{(1)}],\ldots, [u^{(k)}] \in L^p(I,\sR^N) \},
\\
& W^{k,\infty}(I,\sR^N) = \{ [u] \ | \ [u], [u^{(1)}],\ldots, [u^{(k)}] \in L^\infty(I,\sR^N) \} ,
\\
& W^{k,p}_{\loc}(I,\sR^N) = \{ [u] \ | \ [u], [u^{(1)}],\ldots, [u^{(k)}] \in L^p_\loc(I,\sR^N) \},
\\
& W^{k,\infty}_{\loc}(I,\sR^N) = \{ [u] \ | \ [u], [u^{(1)}],\ldots, [u^{(k)}] \in L^\infty_\loc(I,\sR^N) \}.
\end{align*}
In what follows, we adopt the common practice of omitting the square brackets, identifying the class of functions equal a.e. with its representative. We also denote $W^{0,p}(I,\sR^N) \vc L^p(I,\sR^N)$ for any $p \in [1,\infty]$ and similarly for the local spaces.
\end{defn}

Observe that we slightly modify the standard definitions of Sobolev spaces, which would require $I$ to be an \emph{open} interval. However, since we explicitly demand the test functions in the definition of the weak derivative to be compactly supported within $I^\circ$, and the elements of Sobolev spaces are defined only up to $\cL$-null subsets, this modification is purely notational and motivated by our desire to emphasize what the original interval $I$ is. 

Notice that, by the H\"older inequality, for any $p_1, p_2 \in [1,\infty]$
\begin{align}
\label{Holder}
p_1 < p_2 \quad \Rightarrow \quad W^{k,p_1}_{\loc}(I,\sR^N) \supset W^{k,p_2}_{\loc}(I,\sR^N).
\end{align}

The spaces $W^{k,p}(I,\sR^N)$, $p \in [1, \infty]$ are Banach spaces with the \emph{$W^{k,p}$-norms} given by
\begin{align*}
& \|u\|_{W^{k,p}(I)} \vc \left( \sum_{l=0}^k \int_I |u^{(l)}(t)|^p dt \right)^{1/p} \ \textup{ for } p < \infty 
\\
\textup{ and } \quad & \|u\|_{W^{k,\infty}(I)} \vc \max_{l \leq k} \textup{ess}\sup\limits_{t \in I} |u^{(l)}(t)|,
\end{align*}
where the essential supremum is defined as the infimum of the set of the essential upper bounds (i.e. the numbers $c$ appearing in the definition of essential boundedness above).

One traditionally denotes $H^k(I,\sR^N) \vc W^{k,2}(I,\sR^N)$, what is motivated by the Hilbert space structure possessed by these Sobolev spaces. Indeed, the norms are in this case induced by the inner product
\begin{align*}
\langle u,v \rangle_{H^k(I)} \vc \sum_{l=0}^k \int_I u^{(l)}(t) \cdot v^{(l)}(t) \, dt.
\end{align*}

However, the norm topology, or the \emph{strong} topology, is not the only one we shall be considering on the Sobolev spaces with $p=2$.
\begin{defn}
The \emph{weak} topology on $H^k(I,\sR^N)$ is the locally convex topology induced by the seminorms $u \mapsto |\langle u,v \rangle_{H^k(I)}|$ indexed by $v \in H^k(I,\sR^N)$. Accordingly, a net $(u_\lambda) \subset H^k(I,\sR^N)$ \emph{weakly converges} to $u$, which is written symbolically as $u_\lambda \rightharpoonup u$, iff $\langle u_\lambda,v \rangle_{H^k(I)} \rightarrow \langle u,v \rangle_{H^k(I)}$ for every $v \in H^k(I,\sR^N)$.
\end{defn}

Just as the names suggest, strong topology is finer than weak topology. Moreover, both these topologies can be ``pulled back'' on the local Sobolev spaces with the help of the restriction maps.
\begin{defn}
\label{loctopdef}
The strong (resp. weak) topology on $H^k_\loc(I,\sR^N)$ is the coarsest topology such that for each $(a,b) \Subset I$ the restriction map $|_{[a,b]}: H^k_\loc(I,\sR^N) \rightarrow H^k([a,b],\sR^N)$ is strongly (resp. weakly) continuous. Accordingly, a net $(u_\lambda) \subset H^k_\loc(I,\sR^N)$ converges \emph{strongly} (resp. \emph{weakly}) to $u$ iff $u_\lambda|_{[a,b]} \rightarrow u|_{[a,b]}$ (resp. $u_\lambda|_{[a,b]} \rightharpoonup u|_{[a,b]}$) in $H^k([a,b],\sR^N)$ for every $(a,b) \Subset I$.
\end{defn}

Notice that the restriction map commutes with the weak differentiation: $(u|_{[a,b]})' = u'|_{[a,b]}$ a.e., provided any (and hence both) of the two terms exist. Indeed, for any $\phi \in C^\infty_\cc((a,b),\sR^N)$ we can write (suppressing $t$)
\begin{align*}
\int_a^b (u|_{[a,b]})' \cdot \phi \, dt = -\int_a^b u|_{[a,b]} \cdot \phi' dt = - \int_a^b u \cdot \phi' dt = \int_a^b u' \cdot \phi \, dt = \int_a^b u'|_{[a,b]} \cdot \phi \, dt.
\end{align*}
Thanks to that, we can write $\|u\|_{H^k[a,b]}$ instead of the cumbersome $\|u|_{[a,b]}\|_{H^k[a,b]}$ without ambiguity. 

Both topologies on $H^k_\loc(I,\sR^N)$ defined above are locally convex topologies. The seminorms inducing the strong topology are of the form $u \mapsto \|u\|_{H^k[a,b]}$ for $(a,b) \Subset I$, whereas the seminorms inducing the weak topology are of the form $u \mapsto |\langle u,v \rangle_{H^k[a,b]}|$ and are indexed by $(a,b) \Subset I$ and $v \in H^k([a,b],\sR^N)$.

There are obvious (set-theoretic) inclusions summarized in the following commutative diagram for any $k \in \sN$ and any $[a,b] \subset I^\circ$.
\begin{center}
\begin{tikzcd}
H^{k+1}(I,\sR^N) \arrow[hookrightarrow]{d}{\iota_2} \arrow[hookrightarrow]{r}{\iota_1} & H^k(I,\sR^N) \arrow[hookrightarrow]{d}{\iota_3}
\\
H^{k+1}_\loc(I,\sR^N) \arrow[rightarrow]{d}{|_{[a,b]}} \arrow[hookrightarrow]{r}{\iota_4} & H^k_\loc(I,\sR^N)\arrow[rightarrow]{d}{|_{[a,b]}}
\\
H^{k+1}([a,b],\sR^N) \arrow[hookrightarrow]{r}{\iota_5} & H^k([a,b],\sR^N)
\end{tikzcd}
\end{center}
\begin{prop}
\label{prop_inclusions}
All the above inclusions between Sobolev spaces are both strongly and weakly continuous.
\end{prop}
\begin{proof}
Since inclusions are of course linear maps, their strong continuity follows here from the (obvious) inequalities between the norms and seminorms:
\begin{align*}
\|.\|_{H^{k+1}(I)} \geq \|.\|_{H^k(I)} \geq \|.\|_{H^k[a,b]} \leq \|.\|_{H^{k+1}[a,b]}.
\end{align*}

In order to show weak continuity, recall that for linear operators between normed spaces strong and weak continuity are equivalent \cite[Theorem 6.17]{Aliprantis}. This immediately implies the weak continuity of $\iota_1$ and $\iota_5$ above, as well as of $|_{[a,b]} \circ \iota_2$ and $|_{[a,b]} \circ \iota_3$. By the arbitrariness of $[a,b] \subset I^\circ$, the latter in turn imply the weak continuity of $\iota_2$ and $\iota_3$. Similarly, the weak continuity of $\iota_4$ follows from the strong, and hence weak continuity of $|_{[a,b]} \circ \iota_4 = \iota_5 \circ |_{[a,b]}$ for any $[a,b] \subset I^\circ$.
\end{proof}

Weak differentiation is another linear map which is continuous in any of the situations considered.
\begin{lem}
\label{lem_diff_cont}
The map $u \mapsto u'$ is both strongly and weakly continuous when considered either on $H^k(I,\sR^N)$ or $H^k_\loc(I,\sR^N)$, $k \geq 1$.
\end{lem}
\begin{proof}
Consider first $': H^k(I,\sR^N) \rightarrow H^{k-1}(I,\sR^N)$. Strong continuity is obvious, which by \cite[Theorem 6.17]{Aliprantis} means it is weakly continuous as well.

Consider now the commutative diagram, where $[a,b] \subset I^\circ$ is arbitrary
\begin{center}
\begin{tikzcd}
H^{k}_\loc(I,\sR^N) \arrow[rightarrow]{d}{|_{[a,b]}} \arrow[rightarrow]{r}{{}^\prime} & H^{k-1}_\loc(I,\sR^N)\arrow[rightarrow]{d}{|_{[a,b]}}
\\
H^{k}([a,b],\sR^N) \arrow[rightarrow]{r}{{}^\prime} & H^{k-1}([a,b],\sR^N)
\end{tikzcd}
\end{center}
By the previous part of the proof, the map $|_{[a,b]} \circ {}^\prime = {}^\prime \circ |_{[a,b]}$ is weakly continuous, which by the arbitrariness of $[a,b]$ yields the weak continuity of the map $': H^k_\loc(I,\sR^N) \rightarrow H^{k-1}_\loc(I,\sR^N)$.
\end{proof}

The following lemma demonstrates that the \emph{distributional convergence} is in a sense the ``common denominator'' mode of convergence for the Sobolev spaces considered.
\begin{lem}
\label{lem_dist_conv}
For any $\phi \in C^\infty_\cc(I^\circ,\sR^N)$ the map $u \mapsto \int_I u(t) \cdot \phi(t) \, dt$ is both strongly and weakly continuous when considered either on $H^k(I,\sR^N)$ or on $H^k_\loc(I,\sR^N)$.
\end{lem}
\begin{proof}
Let $[a,b] \supset \supp \phi$. The considered map is a bounded linear functional on $H^k(I,\sR^N)$, because
\begin{align}
\label{lem_dist_conv1}
\left| \int_I u \cdot \phi \, dt \right| = \left| \int_a^b u \cdot \phi \, dt \right| \leq \|\phi\|_{L^2[a,b]} \|u\|_{L^2[a,b]} \leq \|\phi\|_{L^2[a,b]} \|u\|_{H^k(I)},
\end{align}
where we have used the H\"older inequality and $t$ has been suppressed for clarity. The map is thus strongly, and hence weakly continuous (we again employ \cite[Theorem 6.17]{Aliprantis}).

In case of $H^k_\loc(I,\sR^N)$, simply observe that the considered map factorizes through the restriction map $|_{[a,b]}: H^k_\loc(I,\sR^N) \rightarrow H^k([a,b],\sR^N)$, which is continuous by definition of the (strong or weak) topology on $H^k_\loc(I,\sR^N)$, so it reduces to the first part of the proof.
\end{proof}

The weak topologies on $H^k(I,\sR^N)$ and $H^k_\loc(I,\sR^N)$ are in a sense very close to each other, as attested by the next lemma.
\begin{lem}
\label{lem_bounded}
Let $\A$ be a $\|.\|_{H^k(I)}$-bounded subset of $H^k_\loc(I,\sR^N)$. Then in fact $\A \subset H^k(I,\sR^N)$ and the weak topologies induced on $\A$ from $H^k_\loc(I,\sR^N)$ and $H^k(I,\sR^N)$ coincide.
\end{lem}
\begin{proof}
The fact that $\A \subset H^k(I,\sR^N)$ follows trivially from the definition of Sobolev spaces. In order to show the coincidence of the weak topologies, let $(u_\lambda) \subset \A$ be a net and let also $u \in \A$.

Let us first prove that $u_\lambda \rightharpoonup u$ in $H^k_\loc(I,\sR^N)$ implies that $u_\lambda \rightharpoonup u$ in $H^k(I,\sR^N)$. To this end, observe that any subnet $(u_\xi)$ of $(u_\lambda)$, being $\|.\|_{H^k(I)}$-bounded, has a subsubnet $(u_\kappa)$ weakly convergent to some $\tilde{u} \in H^k(I,\sR^N)$ (this is a consequence of Alaoglu's theorem \cite[Theorem 6.21]{Aliprantis}). It thus suffices to prove that $u = \tilde{u}$ a.e. This can be done by employing the distributional convergence, since by Lemma \ref{lem_dist_conv} we have that, for any $\phi \in C_\cc^\infty(I^\circ,\sR^N)$,
\begin{align*}
\int_I u_\kappa \cdot \phi \, dt \rightarrow \int_I u \cdot \phi \, dt, \ \quad \textnormal{but also} \ \quad \int_I u_\kappa \cdot \phi \, dt \rightarrow \int_I \tilde{u} \cdot \phi \, dt.
\end{align*}
Hence $\int_I (\tilde{u} - u) \cdot \phi \, dt = 0$, what by the arbitrariness of $\phi$ implies that $u = \tilde{u}$ a.e. on $I$.

The converse reasoning is very similar. Assume that $u_\lambda \rightharpoonup u$ in $H^k(I,\sR^N)$. In order to prove that $u_\lambda \rightharpoonup u$ in $H^k_\loc(I,\sR^N)$, fix any $[a,b] \subset I^\circ$ and notice that any subnet $(u_\xi)$ is $\|.\|_{H^k[a,b]}$-bounded (because $\|.\|_{H^k[a,b]} \leq \|.\|_{H^k(I)}$), and hence it has a subsubnet $(u_\kappa)$ weakly convergent to some $\tilde{u} \in H^k([a,b],\sR^N)$. Employing again the distributional convergence, we obtain that $u = \tilde{u}$ a.e. on $[a,b]$, and so $u_\lambda \rightharpoonup u$ in $H^k([a,b],\sR^N)$. Invoking the arbitrariness of $[a,b]$ concludes the proof.
\end{proof}

The next, very important lemma shows that for $k \geq 1$, the elements of $H^k(I,\sR^N)$ and $H^k_\loc(I,\sR^N)$ have continuous representatives, and so can be conveniently regarded as continuous functions.

\begin{lem}
\label{lem_cont_rep}
If $k \geq 1$, then 
\begin{enumerate}[i)]
\item $H^k(I,\sR^N) \subset C(\bar{I},\sR^N)$, and if $u_n \rightharpoonup u$ in $H^k(I,\sR^N)$, then $u_n$ converges to $u$ uniformly on compact subsets of $\bar{I}$.
\item $H^k_\loc(I,\sR^N) \subset C(I^\circ,\sR^N)$, and if $u_n \rightharpoonup u$ in $H^k_\loc(I,\sR^N)$, then $u_n$ converges to $u$ uniformly on compact subsets of $I^\circ$.
\end{enumerate}
\end{lem}
\begin{proof}
Let $u \in H^k(I,\sR^N) \subset H^1(I,\sR^N)$. Then it is well known that $u$ has an absolutely continuous representative on $\bar{I}$ (which we denote by the same symbol $u$) satisfying \cite[Section 5.2, Theorem 2]{Evans}
\begin{align}
\label{ac}
u(t_2) - u(t_1) = \int_{t_1}^{t_2} u'(t)dt
\end{align}
for all $t_1,t_2 \in \bar{I}$. Hence $u \in C(\bar{I},\sR^N)$.

Let now $(u_n) \subset H^k(I,\sR^N) \subset C(\bar{I},\sR^N)$ be a sequence weakly convergent to $u$ and let $(u_m)$ be any its subsequence. Our aim is to construct its subsubsequence convergent to $u$ uniformly on compact sets of $\bar{I}$.

Begin by noticing that $(u_m)$ is $\|.\|_{H^k(I)}$-bounded, being weakly convergent. Moreover, for any $[a,b] \subset \bar{I}$ we have the following estimate on the strength of Morrey's inequality (cf. \cite[Section 5.6, Theorem 5]{Evans}; in fact, it can also be obtained directly from (\ref{ac}))
\begin{align}
\label{Morrey}
\|u_m\|_{C^{0,1/2}[a,b]} \leq M \| u_m \|_{H^1[a,b]} \leq M \| u_m \|_{H^k(I)},
\end{align}
where the constant $M > 0$ depends only on $[a,b]$ and $N$, whereas $\|.\|_{C^{0,1/2}[a,b]}$ denotes the H\"older $\tfrac{1}{2}$-norm
\begin{align*}
\|u_m\|_{C^{0,1/2}[a,b]} \vc \|u_m\|_{C[a,b]} + \sup\limits_{t_1 \neq t_2 \in [a,b]} \frac{|u_m(t_1) - u_m(t_2)|}{|t_1-t_2|^{1/2}},
\end{align*}
with $\|.\|_{C[a,b]}$ denoting the supremum norm. Altogether, we thus obtain that $(u_m)$ is $\|.\|_{C^{0,1/2}[a,b]}$-bounded, which by the Arzel\`a--Ascoli theorem means it possesses a $\|.\|_{C[a,b]}$-convergent subsequence $(u^{(0)}_m)$.

This, of course, is not yet the subsubsequence we are looking for --- we need to somehow `unfix' the compact subinterval $[a,b]$. To this end, let 
$\{[a_j,b_j]\}$ be an exhaustion of $\bar{I}$ by compact intervals. Let $(u^{(1)}_m)$ be a $\|.\|_{C[a_1,b_1]}$-convergent subsequence of $(u^{(0)}_m)$, existing by the above reasoning. Similarly, define inductively $(u^{(j+1)}_m)$ as a $\|.\|_{C[a_{j+1},b_{j+1}]}$-convergent subsequence of $(u^{(j)}_m)$. Clearly, the ``diagonal'' $(u^{(m)}_m)$ is a subsubsequence of $(u_m)$ that converges uniformly on $[a_j,b_j]$ for every $j \in \sN$, and hence on every compact subset of $\bar{I}$.

Denote the limit of the ``diagonal'' subsubsequence by $\tilde{u} \in C(\bar{I},\sR^N)$. To see that $\tilde{u} = u$, notice that both the weak convergence in $H^k(I,\sR^N)$ and the uniform convergence on compact subsets of $\bar{I}$ imply the distributional convergence (cf. Lemma \ref{lem_dist_conv}). For any $\phi \in C_\cc^\infty(I^\circ,\sR^N)$ we thus have
\begin{align*}
\int_I u^{(m)}_m \cdot \phi \, dt \rightarrow \int_I u \cdot \phi \, dt, \quad \textup{but also} \quad \int_I u^{(m)}_m \cdot \phi \, dt \rightarrow \int_I \tilde{u} \cdot \phi \, dt.
\end{align*}
But this means that $\int_I (\tilde{u} - u) \cdot \phi \, dt = 0$ for all $\phi \in C_\cc^\infty(I^\circ,\sR^N)$, what together with the continuity of both $\tilde{u}$ and $u$ implies that they are equal on $\bar{I}$.


The proof of \emph{ii)} heavily relies on \emph{i)}. Assume this time that $u \in H^k_\loc(I,\sR^N)$, which means $u|_{[a,b]} \in H^k([a,b],\sR^N) \subset C([a,b],\sR^N)$ for any $[a,b] \subset I^\circ$. Hence $u \in C(I^\circ,\sR^N)$.

Let now $(u_n) \in H^k_\loc(I,\sR^N)$ be a sequence weakly convergent to $u$, which means that $u_n|_{[a,b]} \rightharpoonup u|_{[a,b]}$ in $H^k([a,b], \sR^N)$ for any $[a,b] \subset I^\circ$. By \emph{i)}, $(u_n|_{[a,b]})$ converges to $(u|_{[a,b]})$ uniformly on compact subsets of $[a,b]$ and hence on the entire $[a,b]$. By the latter's arbitrariness, this means that $(u_n)$ converges to $u$ uniformly on compact subsets of $I^\circ$.
\end{proof}

\begin{cor}
\label{lem_l2}
If $k \geq 1$, then $u_n \rightharpoonup u$ in $H^k_\loc(I,\sR^N)$ implies $u_n \rightarrow u$ in $L^2_\loc(I,\sR^N)$.
\end{cor}
\begin{proof}
By Lemma \ref{lem_cont_rep}, $u_n$ converges to $u$ uniformly on any $[a,b] \subset I^\circ$, and so
\begin{align*}
\int_a^b |u_n(t) - u(t)|^2 dt \leq (b-a) \|u_n - u\|^2_{C[a,b]} \rightarrow 0.
\end{align*}
\end{proof}

Thanks to Lemma \ref{lem_cont_rep} it makes sense to consider the evaluation map $u \mapsto u(t)$ on the spaces $H^k(I,\sR^N)$ and $H^k_\loc(I,\sR^N)$, provided $k \geq 1$.
\begin{lem}
\label{lem_ev_cont}
If $k \geq 1$, the map $u \mapsto u(t)$ is both strongly and weakly continuous when considered either on $H^k(I,\sR^N)$ (for any $t \in \bar{I}$) or $H^k_\loc(I,\sR^N)$ (for any $t \in I^\circ$).
\end{lem}
\begin{proof}
For any $t \in \bar{I}$ the map $H^k(I,\sR^N) \ni u \mapsto u(t) \in \sR^N$ is well-defined and linear, and by Morrey's inequality one can write that, for any $[a,b] \subset \bar{I}$ such that $[a,b] \ni t$,
\begin{align}
\label{lem_ev_cont1}
|u(t)| \leq \|u\|_{C[a,b]} \leq \|u\|_{C^{0,1/2}[a,b]} \leq M \| u \|_{H^1[a,b]} \leq M \| u \|_{H^k(I)},
\end{align}
which yields its strong (and hence weak) continuity.

As for the map $H^k_\loc(I,\sR^N) \ni u \mapsto u(t)$, where $t \in I^\circ$, it suffices to notice that it factorizes through the restriction map $|_{[a,b]}$ for some $[a,b] \subset I^\circ$ containing $t$, and so it really reduces to the previous case.
\end{proof}

The final lemma explains how the elements of $H^1_\loc$-spaces behave under composition with smooth functions.
\begin{lem}
\label{Psilemma}
Let $\Psi: \sR^N \rightarrow \sR^K$ be a smooth map\footnote{Actually, $C^1$ is enough, cf. \cite[Proposition 4.1.21]{WeberSobolev}.}. Then the map $\Psi_\ast: H^1_\loc(I,\sR^N) \rightarrow H^1_\loc(I,\sR^K)$, $\Psi_\ast(u) \vc \Psi \circ u$ is well defined and has the weak derivative $(\Psi \circ u)' = D\Psi(u) \cdot u'$ with $D\Psi$ denoting the Jacobian matrix of $\Psi$. Moreover, $\Psi_\ast$ is both strongly and weakly sequentially continuous.
\end{lem}
\begin{proof}
Fix $u \in H^1_\loc(I,\sR^N)$, which by Lemma \ref{lem_cont_rep} ii) can be regarded as an element of $C(I^\circ,\sR^N)$. Let us first show that $\Psi \circ u$ and $D\Psi(u) \cdot u'$ are locally $2$-integrable. Indeed, for any $(a,b) \Subset I$ we have that
\begin{align}
\label{Psilemma1}
\int_a^b |\Psi(u(t))|^2 dt & \leq 2 \int_a^b \left(|\Psi(u(t))-\Psi(0)|^2 + |\Psi(0)|^2\right)dt
\\
\nonumber
& \leq 2 L^2_\Psi \int_a^b |u(t)|^2 dt + 2 (b-a)|\Psi(0)|^2 < \infty,
\end{align}
where $L_\Psi$ is the Lipschitz constant of $\Psi|_\K$, where the set $\K \vc u([a,b]) \cup \{0\}$ is compact by the continuity of $u$. We also have that
\begin{align}
\label{Psilemma2}
\int_a^b |D\Psi(u(t)) \cdot u'(t)|^2 dt \leq \max\limits_{x \in \K} |D\Psi(x)|_F^2 \int_a^b |u'(t)|^2 dt < \infty,
\end{align}
where $|D\Psi(x)|_F \vc \sqrt{\sum_{i,j} |D\Psi_{ij}(x)|^2}$ is the Frobenius norm of the Jacobian matrix $D\Psi$ at the point $x$. The map $x \mapsto |D\Psi(x)|_F^2$ is of course continuous and hence bounded on $\K$.

Let us now show that $(\Psi \circ u)'$ exists and equals $D\Psi(u) \cdot u'$. To this end, take $(u_n) \subset C_\cc^\infty(I^\circ, \sR^N)$ such that $u_n \rightarrow u$ in $H^1_\loc(I,\sR^N)$, existing by \cite[Proposition 4.1.13]{WeberSobolev}. Of course, on the approximating sequence we have the standard chain rule
\begin{align*}
(\Psi \circ u_n)' = D\Psi(u_n) \cdot u_n'
\end{align*}
and it suffices to show that $\Psi \circ u_n \rightarrow \Psi \circ u$ and $D\Psi(u_n) \cdot u_n' \rightarrow D\Psi(u) \cdot u'$ in $L^2_\loc(I,\sR^N)$ (cf. \cite[Proposition 4.1.12]{WeberSobolev}). Indeed, for any $(a,b) \Subset I$ we have that, similarly as in (\ref{Psilemma1}),
\begin{align}
\label{Psilemma3}
\int_a^b |\Psi(u_n(t)) - \Psi(u(t))|^2 dt \leq L^2_\Psi \int_a^b |u_n(t) - u(t)|^2 dt \rightarrow 0,
\end{align}
where this time $L_\Psi$ is the Lipschitz constant of $\Psi|_{\K^\ast}$ with $\K^\ast \vc u([a,b]) \cup \bigcup_n u_n([a,b])$ being a compact subset of $\sR^N$ on the strength of Lemma \ref{lem_cont_rep} ii).

One can also write that (suppressing the integration variable $t$)
\begin{align}
\nonumber
& \int_a^b |D\Psi(u_n) \cdot u_n' - D\Psi(u) \cdot u'|^2 dt = \int_a^b |D\Psi(u_n) \cdot (u_n'-u') + (D\Psi(u_n)- D\Psi(u)) \cdot u'|^2 dt
\\
\label{Psilemma4}
& \leq 2 \max\limits_{x \in \K^\ast} |D\Psi(x)|_F^2 \int_a^b |u_n' - u'|^2 dt + 2 \max\limits_{t \in [a,b]} |D\Psi(u_n(t)) - D\Psi(u(t))|_F^2 \int_a^b |u'|^2 dt,
\end{align}
where again the Frobenius norm of the Jacobian matrix appears. The first term clearly tends to zero. As for the second term, observe that we can estimate
\begin{align}
\label{Psilemma5}
|D\Psi(u_n(t)) - D\Psi(u(t))|_F \leq L_{D\Psi} |u_n(t) - u(t)|,
\end{align}
where $L_{D\Psi}$ is the Lipschitz constant of the map $D\Psi: \sR^N \rightarrow \sM_{K \times N}(\sR)$ re\-strict\-ed to $\K^\ast$ (and with $\sM_{K \times N}(\sR)$ endowed with the Frobenius norm). But since by Lemma \ref{lem_cont_rep} ii) $u_n$ converges uniformly to $u$ on $[a,b]$, we have that $\max_{t \in [a,b]} |u_n(t) - u(t)|^2 \rightarrow 0$ and thus the second term in (\ref{Psilemma4}) tends to zero as well. This finishes the proof that the map $\Psi_\ast: H^1_\loc(I,\sR^N) \rightarrow H^1_\loc(I,\sR^K)$, $\Psi_\ast(u) \vc \Psi \circ u$ is well defined and that $(\Psi \circ u)' = D\Psi(u) \cdot u'$ for any $u \in H^1_\loc(I,\sR^N)$.

We now move to proving that $\Psi_\ast$ is both strongly and weakly sequentially continuous. Proving the strong sequential continuity is straightforward --- it amounts to taking any $(u_n) \subset H^1_\loc(I,\sR^N)$ strongly convergent to $u$ and then using estimates (\ref{Psilemma3},\ref{Psilemma4},\ref{Psilemma5}) to prove that $\Psi \circ u_n \rightarrow \Psi \circ u$ in $H^1_\loc(I,\sR^K)$. The proof of the weak sequential continuity, however, requires some more care.

Assume that $(u_n) \subset H^1_\loc(I,\sR^N)$ is weakly convergent to $u$. We want to show that, for any $(a,b) \Subset I$ and any $w \in H^1([a,b],\sR^K)$
\begin{align}
\label{Psilemma6}
& \langle \Psi \circ u_n - \Psi \circ u, w \rangle_{H^1[a,b]}
\\
\nonumber
& = \langle \Psi \circ u_n - \Psi \circ u, w \rangle_{L^2[a,b]} + \langle D\Psi(u_n) \cdot u_n' - D\Psi(u) \cdot u', w' \rangle_{L^2[a,b]} \rightarrow 0.
\end{align}
Indeed, on the strength of Corollary \ref{lem_l2}, $u_n \rightarrow u$ in $L^2_\loc(I,\sR^N)$ and hence, by estimate (\ref{Psilemma3}), we obtain that $\Psi \circ u_n \rightarrow \Psi \circ u$ in $L^2_\loc(I,\sR^K)$ and thus also $\langle \Psi \circ u_n - \Psi \circ u, w \rangle_{L^2[a,b]} \rightarrow 0$ for any $w$.

As for the second term in (\ref{Psilemma6}), one can write, somewhat similarly as in (\ref{Psilemma4}), that
\begin{align}
\nonumber
& |\langle D\Psi(u_n) \cdot u_n' - D\Psi(u) \cdot u', w' \rangle_{L^2[a,b]}| = \left| \int_a^b \left( D\Psi(u_n) \cdot u_n' - D\Psi(u) \cdot u' \right) \cdot w' dt \right|
\\
\label{Psilemma7}
& \leq \left| \int_a^b \left[ \left( D\Psi(u_n) - D\Psi(u) \right) \cdot u_n' \right] \cdot w' dt \right| + \left| \int_a^b \left[ D\Psi(u) \cdot (u_n' - u') \right] \cdot w' dt \right|.
\end{align}
The first term above can be estimated by
\begin{align*}
& \left| \int_a^b \left[ \left( D\Psi(u_n) - D\Psi(u) \right) \cdot u_n' \right] \cdot w' dt \right| \leq \int_a^b |D\Psi(u_n) - D\Psi(u)|_F \, |u_n'| \, |w'| \, dt
\\
& \leq \max\limits_{t \in [a,b]} |D\Psi(u_n(t)) - D\Psi(u(t))|_F \, \|u_n'\|_{L^2[a,b]} \, \|w'\|_{L^2[a,b]}
\\
& \leq L_{D\Psi} \max\limits_{t \in [a,b]} |u_n(t) - u(t)| \, \|u_n'\|_{L^2[a,b]} \, \|w'\|_{L^2[a,b]},
\end{align*}
where we have again employed the Frobenius norm and the Lipschitz constant $L_{D\Psi}$ of the map $D\Psi: \sR^N \rightarrow \sM_{K \times N}(\sR)$ restricted to the compact set $\K^\ast \vc u([a,b]) \cup \bigcup_n u_n([a,b])$, cf. (\ref{Psilemma5}) above. But $\max_{t \in [a,b]} |u_n(t) - u(t)| \rightarrow 0$ (by Lemma \ref{lem_cont_rep} ii)) and $\|u_n'\|_{L^2[a,b]} \leq \|u_n\|_{H^1[a,b]}$ is bounded (because weakly convergent sequences are norm-bounded). Hence the first term on the right-hand side of (\ref{Psilemma7}) tends to zero.

The second term in (\ref{Psilemma7}) is nothing but $|\langle u_n' - u', D\Psi(u)^T \cdot w' \rangle_{L^2[a,b]}|$, which tends to zero on the strength of Lemma \ref{lem_diff_cont}, because $D\Psi(u)^T \cdot w' \in L^2_\loc(I,\sR^K)$ as all the entries of the (transposed) Jacobian matrix belong to $C(I^\circ)$.

This finishes the proof of (\ref{Psilemma6}) and of the entire lemma.
\end{proof}

\section*{Acknowledgments}

The author wishes to thank Micha\l{} Eckstein and Przemys\l{}aw G\'{o}rka for enlightening discussions and insightful comments. The author is also grateful to the anonymous Referees, whose remarks and suggestions lead to notable improvements in various parts of the exposition. The work was supported by the National Science Centre in Poland under the research grant Sonatina (2017/24/C/ST2/00322).

\bibliographystyle{habbrv}
\bibliography{h1CAG_bib}{}

\end{document}